\documentclass[11pt]{article}
\usepackage{mathpazo}
\usepackage{times}

\usepackage{amsthm}
\usepackage{amsmath}
\usepackage{amsfonts}
\usepackage{amssymb}

\usepackage{mathrsfs}  
\usepackage{cite}      
 
\usepackage{upref}

\usepackage{graphicx}

\usepackage{setspace}
\usepackage{multirow}
\usepackage{color}
\usepackage{array}
\usepackage{url}
\usepackage{comment}
\usepackage{enumerate}
\usepackage{hyperref}
\hypersetup{
	colorlinks,
	linkcolor={red!50!black},
	citecolor={green!50!black},
	urlcolor={blue!80!black}
}

\usepackage[lofdepth, lotdepth]{subfig}

\usepackage{authblk}
\usepackage{algorithm}
\usepackage[noend]{algorithmic}
\usepackage[table]{xcolor}
\usepackage{blkarray}

\usepackage{stmaryrd}

\usepackage[top=0.72in, bottom=1in, left=1in, right=1in]{geometry}

\usepackage{caption}
\usepackage{tikz}
\usetikzlibrary{arrows, patterns, shapes.arrows, decorations.pathmorphing, matrix, calc, decorations.pathreplacing, arrows.meta,automata,positioning}

\pgfkeys{tikz/mymatrixenv/.style={decoration=brace,every left delimiter/.style={xshift=3pt},every right delimiter/.style={xshift=-3pt}}}
\pgfkeys{tikz/mymatrix/.style={matrix of math nodes,left delimiter=[,right delimiter={]},inner sep=2pt,column sep=1em,row sep=0.5em,nodes={inner sep=0pt}}}
\pgfkeys{tikz/mymatrixbrace/.style={decorate,thick}}
\newcommand\mymatrixbraceoffseth{0.5em}

\newcommand*\mymatrixbraceleft[4][m]{
	\draw[mymatrixbrace] ($(#1.north east)!(#1-#2-1.north east)!(#1.south east)+(\mymatrixbraceoffseth,0)$)
	-- node[right=2pt] {#4} 
	($(#1.north east)!(#1-#3-1.south east)!(#1.south east)+(\mymatrixbraceoffseth,0)$);
}

\theoremstyle{plain}
\newtheorem{theorem}{Theorem}

\newtheorem{lemma}[theorem]{Lemma}
\newtheorem{proposition}[theorem]{Proposition}

\theoremstyle{definition}
\newtheorem{definition}[theorem]{Definition}
\newtheorem{example}{Example}

\newtheorem{remark}[example]{Remark}

\newcommand{\Lref}[1]{Lem\-ma\,\ref{#1}}
\newcommand{\Cref}[1]{Co\-ro\-lla\-ry\,\ref{#1}}
\newcommand{\Dref}[1]{Def\i\-ni\-tion\,\ref{#1}}
\newcommand{\Sref}[1]{Sec\-tion\,\ref{#1}}
\newcommand{\Tbref}[1]{Tab\-le\,\ref{#1}}
\newcommand{\Fref}[1]{Fig\-ure\,\ref{#1}}

\newcommand{\XX}{\mathbb{X}}

\DeclareMathAlphabet{\mathbfsl}{OT1}{ppl}{b}{it} 

\newcommand{\vr}{\mathbfsl{r}}
\newcommand{\va}{\mathbfsl{a}}
\newcommand{\vb}{\mathbfsl{b}}
\newcommand{\vd}{\mathbfsl{d}}
\newcommand{\vm}{\mathbfsl{m}}

\newcommand{\vp}{\mathbfsl{p}}

\newcommand{\vx}{\mathbfsl{x}}

\newcommand{\vzero}{\mathbf{0}}
\newcommand{\vc}{\mathbfsl{c}}
\newcommand{\vA}{\mathbfsl{A}}
\newcommand{\vB}{\mathbfsl{B}}
\newcommand{\vC}{\mathbfsl{C}}
\newcommand{\vD}{\mathbfsl{D}}

\newcommand{\cS}{\mathcal{S}}

\newcommand{\cC}{\mathcal{C}}

\newcommand{\cE}{\mathcal{E}}
\newcommand{\cF}{\mathcal{F}}

\newcommand{\cX}{\mathcal{X}}

\newcommand{\cP}{\mathcal{P}}
\newcommand{\cT}{\mathcal{T}}
\newcommand{\sS}{\mathbb{S}}

\newcommand{\vdelta}{{\pmb{\delta}}}

\newcommand{\vsigma}{{\pmb{\sigma}}}
\newcommand{\vtau}{{\pmb{\tau}}}
\newcommand{\vrho}{{\pmb{\rho}}}
\newcommand{\vlambda}{{\pmb{\lambda}}}
\newcommand{\vmu}{{\pmb{\mu}}}

\renewcommand{\le}{\leqslant}

\renewcommand{\ge}{\geqslant}

\newcommand{\be}[1]{\begin{equation}\label{#1}}
	\newcommand{\ee}{\end{equation}}
\newcommand{\eq}[1]{(\ref{#1})}

\newcommand{\per}{{\rm per}}

\newcommand{\prob}[1]{{\rm Prob}{\left(#1\right)}}
\newcommand{\floor}[1]{{\left\lfloor #1\right\rfloor}}
\newcommand{\ceil}[1]{{\left\lceil #1\right\rceil}}

\newcommand{\Strut}[2]{\rule[-#2]{0cm}{#1}}

\makeatletter

\gdef\@punct{.\ \ }  
\def\@sect#1#2#3#4#5#6[#7]#8{%
	\ifnum #2>\c@secnumdepth
	\def\@svsec{}
	\else
	\refstepcounter{#1}\edef\@svsec{%
		\ifnum #2>0{{\csname the#1\endcsname}}.\fi%
		\hskip .5em}
	\fi
	\@tempskipa #5\relax
	\ifdim \@tempskipa>\z@
	\begingroup #6\relax
	\@hangfrom{\hskip #3\relax\@svsec}{\interlinepenalty \@M #8\par}
	\endgroup
	\csname #1mark\endcsname{#7}
	\addcontentsline{toc}{#1}{\ifnum #2>\c@secnumdepth\else
		\protect\numberline{\csname the#1\endcsname}\fi#7}
	\else
	\def\@svsechd{#6\hskip #3\@svsec #8\@punct\csname #1mark\endcsname{#7}
		\addcontentsline{toc}{#1}{\ifnum #2>\c@secnumdepth \else
			\protect\numberline{\csname the#1\endcsname}\fi#7}}
	\fi
	\@xsect{#5}}

\def\@ssect#1#2#3#4#5{\@tempskipa #3\relax
	\ifdim \@tempskipa>\z@
	\begingroup #4\@hangfrom{\hskip #1}{\interlinepenalty \@M #5\par}\endgroup
	\else \def\@svsechd{#4\hskip #1\relax #5\@punct}\fi
	\@xsect{#3}}

\makeatother

\title{\huge Computing Permanents on a Trellis}

\author[1]{\large Han Mao Kiah}
\affil[1]{\small School of Physical and Mathematical Sciences, Nanyang Technological University, Singapore}
\author[2]{\large Alexander Vardy}
\author[2]{\large Hanwen Yao}
\affil[2]{\small Department of Electrical \& Computer Engineering, University of California San Diego, LA Jolla, CA, USA}

\begin{document}
\date{\today}

\maketitle
\thispagestyle{empty}

\begin{abstract}
\noindent
Computing the permanent of a matrix is 
a classical problem that attracted considerable interest 
since the work of Ryser (1963) and Valiant (1979). 
A trellis $T$ is an edge-labeled directed graph with~the~property 
that every vertex in $T$ has a well-defined depth; trellises were
extensively studied in coding theory since the 1960s. 
In this work, we establish a connection between the two domains.
We introduce the canonical trellis $\cT_n$, which represents the set
of permutations in the symmetric group $\sS_n$, and show that
the permanent of an arbitrary $n \times n$ matrix $\vA$ can be computed
as a flow on this trellis.
Under appropriate normalization, such trellis-based computation
invokes slightly \emph{less} additions and multiplications than the 
currently \emph{best known} methods for exact computation of the permanent.
Moreover, if the matrix $\vA$ has structure, the canonical trellis $\cT_n$
may become amenable to vertex merging, thereby significantly reducing
the complexity of the computation. Herein, we consider the following
special cases. 
\begin{description}
	\item[Repeated rows:]
	Suppose $\vA$ has only $t<n$ distinct rows, where $t$ is a constant.
	This case is of importance in boson sampling. The best known method 
	to compute $\per(A)$ in this case, due~to~Clifford and Clifford (2020),
	has running time $O(n^{t+1})$. Merging vertices in $\cT_n$, we obtain 
	a reduced trellis that improves upon this result of Clifford and Clifford 
	by a factor of about $n/t$.
	\item[Order statistics:]
	Using trellises, we compute the joint distribution of $t$ order statistics 
	of $n$ independent, but not identically distributed, random variables
	$X_1,\ldots,X_n$ in time $O(n^{t+1})$.\ Previously,~poly\-nomial-time
	methods were known only for the case where $X_1,\ldots,X_n$ are drawn
	from at most two non-identical distributions. For this case,
	we reduce the time complexity from $O(n^{2t})$ to $O(n^{t+1})$.
	\item[Sparse matrices:] 
	Suppose that 
	each entry in $\vA$ is nonzero with probability $d/n$, where $d$ is constant.
	We show that in this case, the canonical trellis $\cT_n$ can be pruned to 
	exponentially fewer vertices. The resulting running time is $O(\phi^n)$,
	where $\phi$ is a constant strictly less than $2$.
	\item[TSP distances:]  
	Intersecting the canonical trellis $\cT_n$ with another trellis that
	represents walks on a complete graph, we obtain a trellis that
	represents circular permutations. Using the latter trellis to solve
	the traveling salesperson problem recovers the well-known Held-Karp
	algorithm.
\end{description}
Notably, in all these cases, the reduced trellis can be obtained using
the standard vertex-merging procedure 
that is well known in the theory of trellises.
We expect that this merging procedure 
and other results from trellis theory
can be applied to many more 
structured matrices of interest.
\end{abstract}

\newpage

\section{Introduction}
\label{sec:intro}

Consider an $n\times n$ matrix 
\smash{$\vA=\big(a_{ij}\big)_{1\le i, j\le n}$}
over a field. The \emph{permanent of $\vA$}, denoted $\per(\vA)$, is defined
by the following expression:\vspace{-0.54ex}
\begin{equation}
	\label{permanent-def}
	\per(\vA)
	\, \triangleq 
	\sum_{\vsigma\in\sS_n} \prod_{i=1}^n a_{i\sigma_i}
\end{equation}
where $\sS_n$ denotes the set of all permutations of
$[n]\triangleq \{1,2,\ldots, n\}$. 
Permanents have numerous applications in combinatorial enumeration,
discrete mathematics, and statistical physics. Recently, the study 
of permanent computation attracted much renewed interest due to its 
connection to boson sampling and the demonstration of the so-called 
quantum supremacy \cite{Aaronson.2011,Clifford.2018,Clifford.2020}.

\looseness=-1
It is well known that exact evaluation of the permanent is computationally 
intractable. Specifically,~Vali\-ant \cite{Valiant.1979} demonstrated in 
1979 that computing the permanent of a $\{0,1\}$-matrix is 
{\large$\#{\tt P}$}-complete.
It is thus as difficult as any problem in {\large${\tt NP}$}.
Indeed, it is this intractability that led 
Aaronson and Arkhipov \cite{Aaronson.2011} to propose boson sampling 
as an attainable experimental demonstration of quantum advantage.
In general, both exact and approximate computation of permanents 
have been an active area of research since at least 1979.

In this work, we focus on {\em exact} methods for computing the permanent.
Such methods are inherently exponential-time, at least for unstructured matrices.
Straightforward evaluation of the expression in \eq{permanent-def}~requires
$n!n$ arithmetic operations. This was improved upon by Ryser~\cite{Ryser.1963}
in 1963, who showed that\vspace{-0.54ex}
\be{Ryser}
\per(\vA) 
\, = \,
(-1)^n \hspace{-0.54ex}
\sum_{S \subseteq [n]} \!(-1)^{|S|} \prod_{j=1}^n \sum_{i \in S}a_{ij}
\vspace{-0.54ex}
\ee
where the outer sum is over all nonempty subsets $S$ of $[n]$.
Ryser's formula \eq{Ryser} is based on the principle of inclusion-exclusion,
and its evaluation involves $O(n^22^n)$ arithmetic operations in Ryser's
original analysis~\cite{Ryser.1963}.
Later, using Gray codes to represent the elements of $S \subseteq [n]$, 
Nijenhuis and Wilf \cite{NW.1978} reduced this
to $O(n2^{n-1})$. In 2010, Glynn \cite{Glynn.2010} gave an alternative
formula for computing~the~permanent, namely:
\be{Glynn}
\per(\vA) 
\, = \,
\frac{1}{2^{n-1}}
\left[
\sum_{\vdelta} \biggl(\prod_{k=1}^n \delta_k\!\biggr)
\prod_{j=1}^n \sum_{i=1}^{n} \delta_i a_{ij}
\right]
\ee
where the outer sum is over all $2^{n-1}$ vectors
$
\vdelta =
(\delta_{1},\delta_{2},\ldots,\delta_{n}) \,{\in}\, \{+1,-1\bigr\}^{n}
$
with $\delta_{1}={+}1$.
Using Gray codes, the complexity of evaluating the Glynn formula \eq{Glynn}
is also $O(n2^{n-1})$.
For structured matrices, 
the com\-plexity can be further reduced. For example, suppose that
$A$ has only $t<n$ distinct rows, where $t$ is a~constant, as is the
case in boson sampling. For this case, it was shown by 
Clifford and Clifford \cite{Clifford.2018,Clifford.2020} that the 
permanent can be computed in $O(n^{t+1})$ time.
Another example of useful structure is sparsity.
Here,~under various sparsity assumptions, several papers
\cite{Servedio.2005, Bjorklund.2012, Lundow.2020}
showed that the permanent can be computed in time 
$O\bigl(n^2 (2-\epsilon)^n\bigr)$, where $\epsilon$ 
is some positive constant.
While the methods of
\cite{Servedio.2005, Bjorklund.2012, Lundow.2020, Clifford.2018,Clifford.2020}
and other papers are considerably faster for structured matrices, 
a key common ingredient 
in all of them is still
either the Ryser formula \eq{Ryser}
or the Glynn formula \eq{Glynn}, or their variants.

\vspace{1.80ex}
\subsection{Our contributions: Computing permanents on a trellis}
\vspace{-0.50ex}

\noindent
In this paper, we propose a very 
different approach to exact permanent computation.
Specifically, we use a graph structure called 
\emph{trellis} to compute permanents. 
One advantage of this approach is that tools and
results from the theory of trellises 
can be 
brought to bear on permanent computation, 
as discussed in what follows.\pagebreak[3.99]

Trellis theory is a well-developed branch of coding theory.
The trellis was invented by Forney \cite{Forney.1967} over 50 years ago
to illustrate the Viterbi decoding algorithm \cite{Viterbi.1967} for
convolutional codes.
It has since been studied extensively by coding theorists; 
see \cite{Vardy.1998} for a survey.
Roughly speaking, a trellis $T$ is an edge-labeled directed graph 
where all paths between two distinguished~vertices 
(called the \emph{root} and the \emph{toor}) have the same length $n$.
Hence, if the edge labels come from an 
alphabet $\Sigma$,
we can regard the set of paths from the root to the toor in $T$
as a code $\cC$ of length $n$ over~$\Sigma$. We then say that $T$
\emph{represents} $\cC$.
Given a code $\cC$, one key objective 
is to find the \emph{minimal trellis} 
for $\cC$ --- that is, a trellis that represents $\cC$ with 
as few vertices and edges as possible. To this end, the authors
of \cite{Kschischang.1996,VardyK.1996}, introduced 
a \emph{vertex merging procedure} that allows one to reduce 
the number of vertices in a trellis while maintaining the set 
of path labels. Furthermore,~it~is shown in \cite{Kschischang.1996,VardyK.1996}
that this vertex merging procedure \emph{always} results in {the unique} 
minimal trellis for $\cC$, provided $\cC$ belongs to a certain large
class of codes known as \emph{rectangular codes}.

Herein, we regard $\sS_n$, the set of permutations of $[n]$, 
as a code of length $n$ over the alphabet $[n]$. Although this code
is clearly nonlinear, it turns out that it is rectangular. Thus we
can use the vertex merging procedure of \cite{Kschischang.1996,VardyK.1996}
to find the unique minimal trellis representation $\cT_n$ for $\sS_n$.
We henceforth refer to $\cT_n$ as the~\emph{canonical trellis}.
With this, the permanent of a general $n \times n$ matrix $\vA$
can be computed as the flow from the root to the toor on the 
canonical trellis, when its edges are re-labeled with the entries in $\vA$.
To~com\-pute this flow, we use the Viterbi algorithm over the 
sum-product semiring (see~\cite[Chapter\,3]{Vardy.1998}).

The resulting computation requires $n(2^{n-1}-1)$ multiplications
and $(n-2)2^{n-1}+1$ additions. This is of the same order as
the number of arithmetic operations required to evaluate
the Ryser formula \eq{Ryser} or the Glynn formula \eq{Glynn}.
However, a more refined analysis, shows that the trellis-based 
approach is slightly better. Using a simple normalization, the number 
of multiplications in the Viterbi algorithm can be reduced to 
$$
n2^{n-1}-\ceil{\frac{n}{2}}\binom{n}{\floor{n/2}}+~n^2-n
\ = \ 
\Bigl(n-\sqrt{2n/\pi}\Bigr)2^{n-1} + O\bigl(n^{-3/2}2^n\bigr)
$$
while the number of additions remains the same.
In contrast, the best-known exact methods, 
namely those~of Nijenhuis-Wilf~\cite{NW.1978} and Glynn~\cite{Glynn.2010},
require $(n-1)2^{n-1}$ multiplications and $(n+1)2^{n-1}+O(n^2)$ additions
(for quick reference, we summarize these and other complexity measures
in \Tbref{table:general}). 
Thus the trellis-based approach introduced herein is \emph{slighly faster}
than any other known method for exactly computing the permanent of general
matrices, 
albeit at the expense of exponential space complexity.

\begin{table*}[!h]
	\begin{center}
		$\,$\\[1.44ex]
		\small
		\renewcommand{\arraystretch}{1.5}
		\begin{tabular}{| l | c | c |}
			\hline
			\bfseries\sffamily\footnotesize Method & 
			\bfseries\sffamily\footnotesize Number of multiplications & 
			\bfseries\sffamily\footnotesize Number of additions  
			\\
			\hline\hline
			Ryser~\cite{Ryser.1963} & 
			$2(n-1)2^{n-1}-(n-1)$ & $(n^2-2n+2)2^{n-1}+n-2$  \\ \hline
			Ryser~\cite{Ryser.1963} with Gray code ordering~ & 
			$(n-1)2^n-(n-1)$ & $2(n+1)2^{n-1}-2(n+1)$   \\ \hline
			Nijenhuis-Wilf~\cite{NW.1978} & 
			$(n-1)2^{n-1}$ & $(n+1)2^{n-1}+n^2-2n-1$  \\ \hline
			Glynn~\cite{Glynn.2010} with Gray code ordering~ & 
			$(n-1)2^{n-1}$ & $(n+1)2^{n-1}+n^2-2n-1$  \\ \hline\hline
			Canonical trellis & 
			$n2^{n-1}-n$ & $(n-2)2^{n-1}+1$ \\ \hline
			Trellis with normalization & 
			$\Strut{0pt}{9pt}n2^{n-1}-\ceil{\frac n2}\binom{n}{\floor{n/2}}+n^2-n$ & 
			$(n-2)2^{n-1}+1$  \\ \hline\hline 
			\multicolumn{3}{c}{}\\[-6.30ex]
		\end{tabular}
	\end{center}
	\caption{Number of arithmetic operations
		required to compute the permanent of general $n\times n$ matrices}
	\label{table:general}
\end{table*}

\looseness=-1
\hspace*{-4.5pt}In addition to the slight improvement in time complexity,  
the trellis-based approach has other merits.~First,
computation on a trellis avoids the overflow problems 
that are characteristic of inclusion-exclusion formulae 
(see, for example, \cite[Chapter 23, page 223]{NW.1978} 
for a discussion). 
%
%
Second, and most~importantly, whenever the matrix $\vA$ has structure, 
the canonical trellis $\cT_n$ may become amenable to vertex merging
and/or pruning. In what follows, we consider several specific 
instances of this circumstance. We point out, however, that the
general principle applies much more broadly: 
vertices in $\cT_n$ 
can be merged whenever they are \emph{mergeable} or pruned whenever
they are \emph{nonessential} (see \Sref{sec:canonical} for the 
definition of these notions),
thereby reducing the complexity of the permanent computation.
We therefore expect~that~results from the theory of trellises
can be applied to many more structured matrices of interest.


\vspace{1.80ex}
\subsection{Our contributions: Structured matrices}
\vspace{-0.50ex}

\noindent
We specifically consider trellis-based computation 
of the permanent in three different scenarios:~matrices
with repeated rows (boson sampling),
computing joint distribution of $t$ order statistics,
and sparse matrices. 
In what follows, we formally state our contributions 
and compare with previously best known~results. 

\subsubsection{Matrices with repeated rows}
\vspace{-0.36ex}

Suppose the $n\times n$ matrix of interest has $t<n$ distinct rows,
where $t$ is a constant. We further assume that these $t$ distinct 
rows appear with multiplicities $m_1,m_2,\ldots,m_t$. The task of
computing the permanent~of such a matrix arises in the context
of boson sampling, as proposed in~\cite{Aaronson.2011}.
Clifford and Clifford \cite{Clifford.2018,Clifford.2020}~used 
generalized Gray codes to provide a faster method of
evaluating a formula due to Shchesnovich \cite[Appendix\,D]{Shchesnovich.2013}. 
The resulting computation requires
\be{Clifford-complexity}
(n-1)\biggl[\,\prod_{i=1}^t(m_i+1)-1\biggr]~
\text{multiplications}
\hspace*{3.60ex}\text{and}\hspace*{3.60ex}
(n+1)\biggl[\,\prod_{i=1}^t(m_i+1)-\,2\biggr]~
\text{additions}\hspace*{1.80ex}
\ee
In \Sref{sec:repeated}, we directly construct the minimal trellis for
computing the permanent of matrices with repeated rows. This trellis has
$|V| = (m_1+1)(m_2+1)\cdots(m_t+1)$ vertices and $|E| \le t|V|$ edges.
The resulting computation 
requires at most 
$t \prod_{i=1}^t(m_i+1)$ multiplications and at most 
$(t-1)\prod_{i=1}^t(m_i+1)$ additions. 
Thus, as compared to \eq{Clifford-complexity}, 
computing the permanent on a trellis reduces the number 
of arithmetic operations by a factor of about $n/t$.
As a consequence, classical algorithms would be 
able to solve the exact boson~sam\-pling problem 
for system sizes beyond what was previously possible.

\vspace{0.90ex}
\subsubsection{Order statistics}
\vspace{-0.36ex}

Suppose we have $n$ independent, but not necessarily identical,
real-valued random variables $X_1,X_2,\ldots,X_n$. We draw one 
sample from each distribution and order them so that 
\smash{$X_{(1)}\le X_{(2)}\le \ldots\le X_{(n)}$}. Further, let us
fix $t$ distinct integers $r_1,r_2,\dots,r_t$ with 
$1\le r_1<r_2<\cdots <r_t\le n$ 
and $t$ real values $x_1,x_2,\dots,x_t$
with $x_1\le x_2\le \cdots \le x_t$. 
Then the task of interest is to evaluate the joint probability 
\smash{$\prob{\bigwedge_{\ell=1}^t X_{(r_\ell)}\le x_\ell}$}. 

It is shown in \cite{Vaughan.1972} 
and \cite{Bapat.1989, Balasubramanian.1991, Balasubramanian.1996}
that this joint probability can be computed by the summing
suitably scaled permanent functions. Assuming that
$t, r_1, r_2,\ldots, r_t$ and $x_1, x_2,\ldots, x_t$ are given constants,
the number~of~permanents in this summation is $\Theta(n^t)$.
However, to the best of our knowledge,
polynomial-time algorithms are available 
only for the case 
where the $n$ random variables are drawn from \emph{at most two} 
distributions~\cite{Glueck.2008}.

In contrast, we show herein that the joint distribution can be computed
efficiently even if 
\emph{all $n$ distributions are distinct}. To this end, we first observe
that each of the $\Theta(n^t)$ permanent functions 
is evaluated on a~matrix with at most $t+1$ distinct rows. 
Now, if we naively invoke $\Theta(n^{t})$ times 
the computation~of~\Sref{sec:repeated},
which evaluates each permanent with complexity $O(n^{t+1})$,
we obtain a running time of $O(n^{2t+1})$. 
However, we can do 
much better.
Rather than of constructing $\Theta(n^t)$ trellises, 
we merge all of them into a~\emph{single trellis} 
with at most $n^{t+1}$ vertices.
We then use an appropriate modification of the Viterbi algorithm 
to evaluate the sum of $\Theta(n^t)$ different flows, all at once,
on the combined trellis.
With this, we are able to compute the joint distribution using at most
$(t+1)n^{t+1}$ multiplications and at most $(t+1)n^{t+1}$ additions.

\vspace{0.90ex}
\subsubsection{Sparse matrices}\label{sec:intro-sparse}
\vspace{-0.36ex}

Another form of matrix structure 
is sparsity. It is known that
for matrices with few nonzero entries,~the~com\-plexity of permanent computation
can be reduced by an exponential factor. 
Over the past decade, this was shown in several papers
\cite{Servedio.2005, Bjorklund.2012, Lundow.2020} 
under various sparsity assumptions. The Ryser formula \eq{Ryser} is
central to the analysis in all these papers.
This formula expresses the permanent as the sum of $2^n-1$ terms,
and each of these terms corresponds to a set $S$ of row indices.
When the matrix $\vA$ is sparse, many of these terms are zero.
Specifically, let \smash{$\vA=\big(a_{ij}\big)_{1\le i, j\le n}$}
and define
$$
\cE
\,\triangleq\,
\Bigl\{
S\subseteq [n]
~:~ \text{there exists $j\in[n]$ such that $a_{ij}=0$ for all $i\in S$}
\Bigr\}
\vspace*{-0.90ex}
$$ 
\looseness=-1
Clearly, those terms in \eq{Ryser} that correspond to $S \,{\in}\, \cE$
need not be evaluated. Let \smash{$\cF = 2^{[n]}\setminus\cE$}
be~the~complement of $\cE$,
and fix an integer $d < n$.
With this notation, \cite{Servedio.2005, Bjorklund.2012, Lundow.2020} 
establish the following bounds on $|\cF|$.
\begin{quote}
	\begin{itemize}
		\item 
		Servedio and Wan \cite{Servedio.2005} showed that if the \emph{total number}
		of nonzero entries in a matrix is at most $dn$, 
		then $|\cF|\le \phi_1^n$, 
		where $\phi_1=2\left(1-1/2^{2d}\right)^{1/8d}$.
		\vspace{0.90ex}
		
		\item 
		Bj\"orklund, Husfeldt, Kaski, and Koivisto \cite{Bjorklund.2012}, 
		showed that if a matrix has at most $d$ nonzero entries \emph{in~every row}, 
		then $|\cF|\le \phi_2^n$, where $\phi_2=(2^d-1)^{1/d}$.
		\vspace{0.90ex}
		
		\item 
		Lundow and Markst\"orm \cite{Lundow.2020}
		showed that if a matrix has at most $d$ nonzero entries 
		\emph{in every row and every column}, 
		then $|\cF|\le \phi_3^n$, where $\phi_3=(2^d-1)^{1/d^2} 2^{1-1/d}$.
	\end{itemize}
\end{quote}
Consequently, for each of these methods, the number of multiplications
required to compute the permanent is at most $(n-1)\phi^n$, 
where $\phi\in\{\phi_1,\phi_2,\phi_3\}$. 
We omit the analysis of the number of additions.
Such~analysis would be quite involved since it depends on the size $|S|$
of the row-index subsets $S \in \cF$.
A more~serious\linebreak problem with the approach of \cite{Bjorklund.2012}
is this: 
it is not clear how the subsets in $\cF$ can be efficiently generated.
In any case, we do not include the complexity of generating $\cF$
in our comparisons (cf.~\Tbref{table:sparse}).

\begin{table*}[!h]
	\begin{center}
		$\,$\\[1.08ex]
		\begin{tabular}{|c |ccccc|c|}
			\hline
			$d$ & 2 & 3 & 4 & 5 & 6 & 
			\begin{tabular}{c}
				\text{\bfseries\sffamily\footnotesize Number of} 
				\\[-0.90ex]
				\text{\bfseries\sffamily\footnotesize \Strut{0ex}{0.72ex}multiplications}
			\end{tabular}
			\\ \hline\hline
			$\phi_1$ & 1.99195  & 1.99869  & 1.99976  & 1.99995  & 1.99999 & 
			$(n-1)\phi_1^n$\Strut{2.70ex}{0ex}\\
			$\phi_2$ & 1.73205  & 1.91293  & 1.96799  & 1.98734  & 1.99476 & 
			$(n-1)\phi_2^n$\\
			$\phi_3$ & 1.86121  & 1.97055  & 1.99195  & 1.99746  & 1.99913 & 
			$(n-1)\phi_3^n$\Strut{0ex}{1.44ex}\\
			\hline
			$\phi_T$ & 1.86466  & 1.95021  & 1.98168  & 1.99326  & 1.99752 & 
			$d\phi_T^n$\Strut{2.70ex}{0ex}\\[0.25ex]
			$\phi_U$ & 1.40255  & 1.63691  & 1.77824  & 1.86430  & 1.91684 & 
			$d\phi_U^n\Strut{0ex}{1.44ex}$
			\\ \hline\hline
			\multicolumn{7}{c}{}\\[-4.50ex]
		\end{tabular}
	\end{center}
	\caption{Number of multiplications required 
		to compute the permanent of sparse matrices}
	\label{table:sparse}
\end{table*}

The trellis-based approach avoids these issues. The key observation is
that when the matrix $\vA$ is sparse, most vertices in the canonical 
trellis $\cT_n$ become \emph{non-essential}, meaning that they do not
lie on any path from the root to the toor. Such vertices 
(and all the edges incident upon them) can be pruned away without 
affecting the flow
from the root to the toor and, hence, the permanent.

Formally, we relax the sparsity~assumptions, adopting a probabilistic
model instead. As before,~fix~an~integer $d<n$, and assume~that~the 
matrix $\vA$ is randomly generated: each entry in $\vA$ is nonzero with 
probability $d/n$ and zero otherwise. For this model, we provide in 
\Sref{sec:sparse} a simple method to prune the canonical trellis $\cT_n$,
and show that the resulting trellis has exponentially fewer vertices
with high probability. We furthermore prove that the {expected}
number of vertices in this trellis is at most 
\be{Un}
U(n) 
\,\triangleq\
\sum_{j=0}^n\sum_{k=0}^j 
(-1)^k\binom{n}{j}\binom{j}{k}
\left( \left(1 - \frac{d}{n}\right)^{\!k} - \left(1 - \frac{d}{n}\right)^{\!j}\,
\right)^j
\le \,
\Bigl(2-e^{-d}\Bigr)^n
\vspace{-0.54ex}
\ee
Though we do not have a closed-form expression for $U(n)$, 
we compute an estimate of $\phi_U \triangleq \lim_{n\to\infty} U(n)^{1/n}$ 
for all $d \le 6$. The resulting constants $\phi_U$ and $\phi_T = 2-e^{-d}$
are compared with $\phi_1,\phi_2,\phi_3$ in \Tbref{table:sparse}. 

\begin{figure*}[!h]
	\begin{center}
		$\,$\\[1.08ex]
		\includegraphics[height=7.74cm]{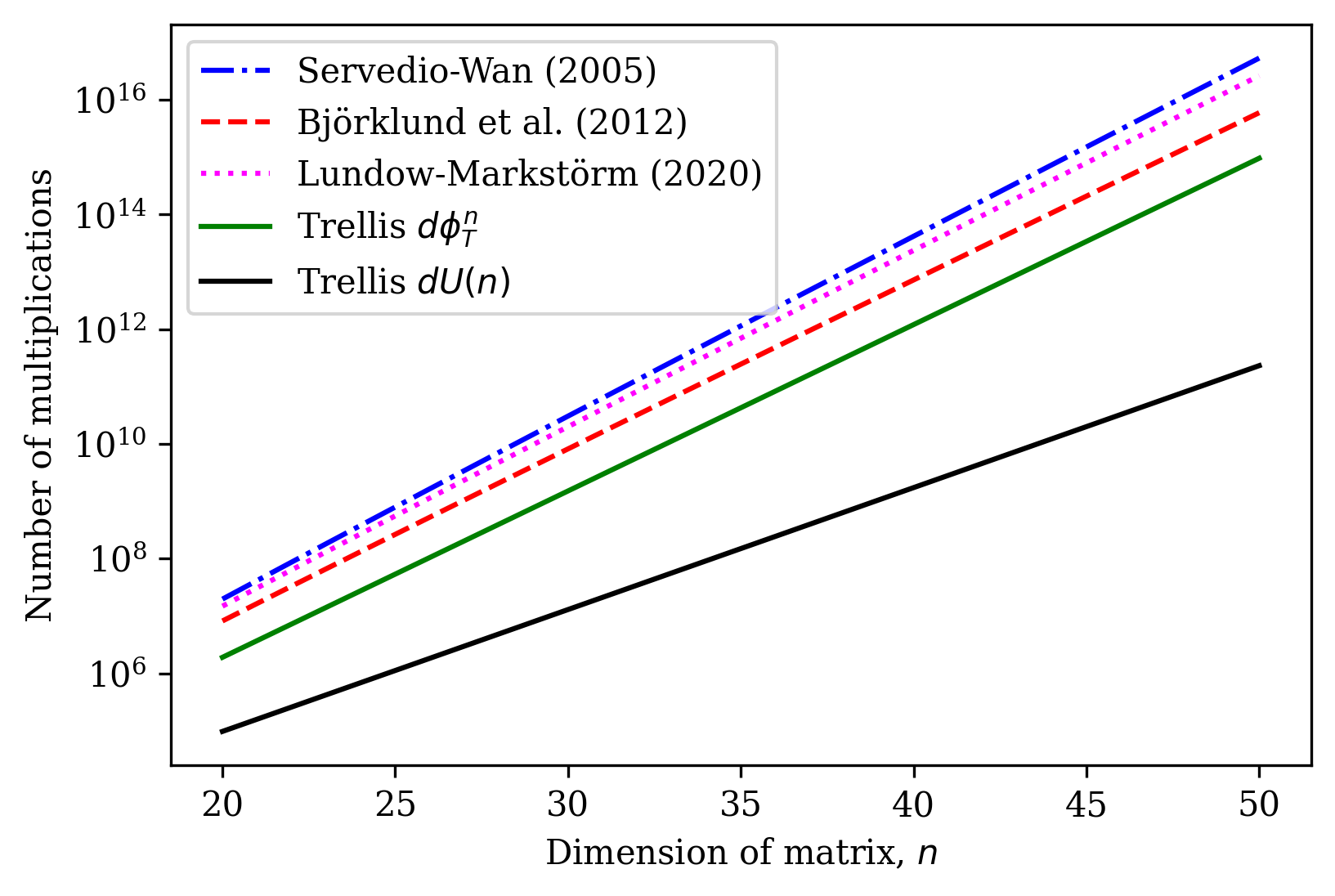}%
		\vspace*{-3.60ex}
	\end{center}
	\caption{Number of multiplications to compute the permanent 
		for sparse matrices of dimension at most 50}
	\label{fig-sparse}
\end{figure*}

For $d=3$ and $n$ up to $50$, we compute the value of $U(n)$ in \eq{Un}
exactly, and the results are plotted in \Fref{fig-sparse}. 
From \Tbref{table:sparse} and \Fref{fig-sparse}, we observe that
--- at least in the probabilistic model --- the trellis-based computation 
is exponentially faster than the known methods.

\vspace{2.70ex}
\subsection{Computing permanent-like functions: Solving the TSP on a trellis}\label{sec:into-tsp}
\vspace{-0.50ex}

We show that trellises can be also used to compute certain ``permanent-like'' 
functions. Specifically,~we~consider the task of evaluating\vspace{-0.54ex}
\begin{equation}
	\label{per*-def}
	\per^*(\vA)
	\, \triangleq 
	\sum_{\vsigma\in\XX} \prod_{i=1}^n a_{i\sigma_i}
\end{equation}
where $\XX$ is a \emph{subset} of the symmetric group $\sS_n$, while the
sum and product operations are over an arbitrary \emph{semiring} 
$(\cS,\cdot,+)$.
The fact that the Viterbi algorithm can be used to compute trellis
flows\footnote{In this paper, we use the term ``flow'' following 
	McEliece~\cite{McEliece.1996}, who gives an excellent exposition 
	of the Viterbi algorithm on trellises. Trellis flows should not be
	confused with \emph{network flows}, as the two are not exactly the same.}
over an arbitrary semi\-ring 
is due to McEliece~\cite{McEliece.1996}.
An important special case is the min-sum semiring, where $+$ is the 
operation of taking the minimum and $\,\cdot\,$ is the ordinary summation.
If we furthermore take $\XX$ to be~the~set 
$\raisebox{-0.36ex}{\huge$\circ$}_n$ of circular permutations 
(a permutation is said to be \emph{circular} 
if its cycle decomposition comprises~exactly one cycle) 
of $[n]$, then \eq{per*-def} becomes\vspace{-2.70ex}
\begin{equation}
	\label{per*-TSP}
	\per^*(\vA)
	\, = \
	{\textstyle\min_{\,\vsigma\in\raisebox{-0.36ex}{\LARGE$\circ$}_n}} \sum_{i=1}^n a_{i\sigma_i}
	\vspace{-.90ex}
\end{equation}
Now, if $\vA$ is the matrix of distances between $n$ cities, then \eq{per*-TSP}
is precisely the length of the shortest traveling salesperson (TSP) tour
visiting every city exactly once \cite{HeldKarp.1962,Bjorklund.2012}.

The remaining problem is to find a trellis representation for the set
$\raisebox{-0.36ex}{\huge$\circ$}_n$ of circular permutations. It turns
out that such a trellis can be obtained as the {intersection} of the
canonical trellis $\cT_n$ with another natural trellis which represents
walks on the complete graph connecting all the cities. We use well known
results from trellis theory~\cite{KoetterV.2003,Kschischang.1996} to 
compute the intersection of these trellises. Curiously, running the
Viterbi algorithm over the min-sum semiring on this intersection 
trellis, we recover the Held-Karp algorithm \cite{HeldKarp.1962} 
which is the best-known exact method for solving the TSP.

\section{Canonical trellis for permanent computation}
\label{sec:canonical}

\vspace{-0.50ex}

In this section, we present the main ingredient in all our computations:
the trellis. First, we formally~define a~trellis and describe how to
compute the permanent with the Viterbi algorithm. Then we provide 
a canonical trellis that represents the set of all permutations, 
and show that the permanent computation on this trellis invokes slightly 
less multiplications and additions than the state-of-the-art methods.

A {\em trellis} $\cT=(V, E, L)$ is an edge-labelled directed graph,
where $V$ is the set of {\em vertices}, $E$ is the~set~of ordered
pairs $(v,v')\in V\times V$, called {\em edges}, and $L$ is the 
{\em edge-labelling function}. Specifically, $L$~is~a~function that maps 
an edge $e\in E$ to a symbol $\sigma$ in $\Sigma$, the {\em label
	alphabet}. In this work, the label alphabet $\Sigma$ will be either
$[n]$ or $\mathbb{F}$, the field that our matrix $\vA$ is defined
upon.

The defining property of a trellis is that the set $V$ of vertices can
be partitioned into $V_0,V_1,\ldots, V_n$ such that every edge
$(v,v')$ begins at $v\in V_{j-1}$ and terminates at $v'\in V_j$ for
some $j\in [n]$. For most of this work, the subsets $V_0$ and $V_n$
are singletons, called the {\em root} and the {\em toor},
respectively.  For each path $\vp$ defined by its edge sequence
$e_1e_2\cdots e_t$, we associate the path with its {\em label string
	$L(\vp)\triangleq L(e_1)L(e_2)\cdots L(e_t)\in \Sigma^t$}.  Then for a
given trellis $\cT$, the {\em multiset} of all paths from the root to
toor is denoted $\cC(\cT)$ and we say~that~$\cT$ is a {\em trellis
	representation} for the collection of words in $\cC(\cT)$. 
~{\bfseries\itshape Notation:} 
we will use ``+'' to denote~a~multiset union of paths. For example, the
collection of strings $00,00,11$ will be written as $2(00)+11$.

\begin{example}
	Set $n=3$. Consider the following trellis $\cT$ with $\Sigma=[n]$.
	
	\begin{center}
		\small
		\begin{tikzpicture}[x=3cm,y=1.5cm]
			\tikzstyle{state}=[rectangle,fill=white,draw,line width=0.8mm]
			\tikzstyle{label}=[fill=white, inner sep=1pt]
			\node[state] at (0,0) (0) {$\varnothing$};
			\node[state] at (1,1) (1) {$1$};
			\node[state] at (1,0) (2) {$2$};
			\node[state] at (1,-1) (3) {$3$};
			\node[state] at (2,1.5) (12) {$12$};
			\node[state] at (2,1.0) (13) {$13$};
			\node[state] at (2,0.4) (21) {$21$};
			\node[state] at (2,-0.4) (23) {$23$};
			\node[state] at (2,-1.0) (31) {$31$};
			\node[state] at (2,-1.5) (32) {$32$};
			\node[state] at (3,0) (123) {$123$};
			
			\path[-] (0) edge  node[label] {1} (1);
			\path[-] (0) edge  node[label] {2} (2);
			\path[-] (0) edge  node[label] {3} (3);
			
			\path[-] (1) edge  node[label] {2} (12);
			\path[-] (1) edge  node[label] {3} (13);
			\path[-] (2) edge  node[label] {1} (21);
			\path[-] (2) edge  node[label] {3} (23);
			\path[-] (3) edge  node[label] {1} (31);
			\path[-] (3) edge  node[label] {2} (32);
			
			\path[-] (12) edge  node[label] {3} (123);
			\path[-] (13) edge  node[label] {2} (123);
			\path[-] (21) edge  node[label] {3} (123);
			\path[-] (23) edge  node[label] {1} (123);
			\path[-] (31) edge  node[label] {2} (123);
			\path[-] (32) edge  node[label] {1} (123);
		\end{tikzpicture}
	\end{center}
	Here $V_0=\{\varnothing\}, V_1=\{1,2,3\}, V_2=\{12,13,21,23,31,32\}$,
	abd $V_3=\{123\}$. 
	There are six paths from $\varnothing$ to $123$ and 
	$\cC(\cT) = \sum_{\vsigma\in \sS_3} \vsigma =\sS_3$. 
	Hence we say that $\cT$ is a trellis representation for $\sS_3$.
\end{example}

\looseness=-1
Herein, we are interested in trellis representations for
$\sS_n$, the set of all permutations, because we can use the Viterbi
algorithm on such a trellis to compute the permanent of a matrix.
The Viterbi algorithm is an application of the dynamic programming
method pioneered by Bellman \cite{Bellman.1957}. It was introduced
by Viterbi \cite{Viterbi.1967} in 1967 to perform maximum-likelihood
decoding of convolutional codes. Here, we describe the
Viterbi algorithm in the context of permanents; a more general
version of the algorithm is described in \Sref{sec:tsp}.
\vspace{1mm}

Consider an $n\times n$ matrix $\vA$, and  
suppose that $\cT=(V,E,L)$ is a trellis representation for $\sS_n$,
with $V_0=\{\rm root\}$ and $V_n=\{\rm toor\}$. 
Then, to compute the permanent of $\vA$, we do the following.
\begin{quote}
	\begin{enumerate}[(1)]
		\item 
		Relabel the edges with $\vA$, and call the labelling
		$L_\vA$. Specifically, if $e$ is an edge from $V_{j-1}$ to $V_j$ and
		$L(e)=i$, then set the label $L_\vA(e)$ to be $a_{ij}$. Call this new
		trellis $\cT(\vA)$.
		\item 
		Perform the Viterbi algorithm on $\cT(\vA)$. 
		For each node $v\in V$, we assign a flow $\mu(v)$ 
		that is computed in the following recursive manner:\\[1mm]
		\hspace*{20mm}Set $\mu({\rm root})=1$\\[1pt]
		\hspace*{20mm}for $j\in [n]$ \\[0pt]
		\hspace*{30mm}for $v\in V_j$ \\[0pt]
		\hspace*{40mm}set $\mu(v) = \sum_{(u,v)\in E}L_\vA(u,v)\mu(u)$
		
		\item 
		Then $\per(\vA)$ is given by $\mu({\rm toor})$.
	\end{enumerate}
\end{quote}

We refer to this procedure as \emph{trellis-based computation} 
of the permanent. Its complexity can be explicitly measured by the
following graph-theoretic quantities:\vspace*{-9mm}

\begin{align}
	\text{number of multiplications} & = |E|-\deg({\rm root}),\label{eq:trellis-multiplies}\\
	\text{number of additions} & = |E|-|V|+1,\label{eq:trellis-additions}\\
	\text{space} & = \max\left\{|V_0|,|V_1|,\ldots, |V_n|\right\}.\label{eq:trellis-space}
\end{align}
\vspace*{-8mm}

\noindent It follows from these expressions that in order to reduce the
complexity, we need to find trellis representations for $\sS_n$ that
use as few vertices and edges as possible. To do so, we look at
{vertex merging}.

\vspace{0.50ex}
\subsection{Minimal trellises and vertex mergeability}

Consider some collection $\cX$ of words of length $n$.
One key objective in the study of trellises in coding theory is to find a ``small'' trellis $\cT$ so that $\cC(\cT)=\cX$. Formally, we say that $\cT^*=(V^*,E^*,L^*)$ is a {\em minimal trellis} for $\cX$ if the following 
holds:\vspace{0.25ex}
\begin{center}
	\centerline{\it\sl
		for all other trellis representations $\cT=(V,E,L)$ of $\cX$, 
		we have $|V^*_j|\le |V_j|$ ~for all~ $j = 0,1,\ldots,n$.}
	\vspace{-3.00ex}
\end{center}
There are examples of word collections that do not admit a minimal
trellis representation.  Nevertheless,~if 
$\cX$ obeys certain properties (cf.~\Dref{def:rectangular}), 
we have that $\cX$ admits a unique minimal
trellis representation. Moreover, there is a simple {\em merging} procedure
that finds this trellis \cite{Kschischang.1996, VardyK.1996}.  Here,
by merging, we refer to a procedure that reduces the number of
vertices and edges in the trellis while preserving the set of
length-$n$ paths in the trellis (henceforth, unless stated otherwise,
a ``set of paths'' also refers to a multiset of paths).

Now, to define our merging procedure, we need to identify when two vertices
can be merged. To this end, we study certain local properties of a
vertex and introduce the notions of ``past'' and ``future'' of a
vertex. Specifically, for a vertex $v$ in the trellis, we define the
following sets:\vspace{-8mm}

\begin{align*}
	\cP(v) & \triangleq 
	~\text{multiset of all label strings of the paths from the root to $v$},
	\\[-1pt]
	\cF(v) & \triangleq 
	~\text{multiset of all label strings of the paths from $v$ to the toor}.
\end{align*}
We refer to $\cP(v)$ and $\cF(v)$ as the {\em past} and {\em future} of $v$ respectively.

\begin{definition}
	Two distinct $v,v'\in V$ are said to be {\em mergeable} if 
	\vspace{-3mm}
	
	\begin{equation}\label{eq:merge}
		\cP(v)\cF(v) +\cP(v')\cF(v')
		=\cP(v)\cF(v')+\cP(v')\cF(v)\,. 
	\end{equation}
	
\end{definition} 
\vspace{-3mm}

Next, we describe the {\em merging process}. Suppose that $v$ and $v'$ are two mergeable vertices in $\cT$ and we want to merge them. We first observe that \eqref{eq:merge} implies that $v$ and $v'$ belong to some $V_j$ for some $j$.
In the new trellis $\cT^*$, we set the vertex set $V^*$ to be $V\setminus\{v,v'\}\cup \{v^*\}$. For the edge set $E^*$, we keep an edge $(w,w')\in E$ and its labels as long as $w\ne v$, $w\ne v'$, $w'\ne v$ and $w'\ne v'$.
\vspace{-3mm}

\begin{itemize}
	\item If both $(w,v)$ and $(w,v')$ are edges, we include the edge $(w,v^*)$ with its label being $L^*(w,v^*)=L(w,v)+L(w,v')$. If $(w,v)$ or $(w,v')\in E$, we include the edge $(w,v^*)$ with $L^*(w,v^*)=L(w,v)$ or $L^*(w,v^*)=L(w,v')$, respectively.\\[-8mm]
	\item If both $(v,w)$ and $(v',w)$ are edges, we include the edge $(v^*,w)$ with the edge-label $L^*(v^*,w)=(L(v,w)+L(v',w))/2$.
	If $(v,w)$ or $(v',w)\in E$, we include the edge $(v^*,w)$ with $L^*(v^*,w)=L(v,w)/2$ or $L^*(v^*,w)=L(v',w)/2$, respectively.
\end{itemize}  
\vspace{-3mm}

Note that we abuse notation by ``adding'' symbols in $\Sigma$ and also ``multiplying" these symbols by rational scalars. We can justify these operations if we regard the multiset of label strings as elements in the semigroup algebra $\mathbb{Q}[\Sigma^+]$.  The technicalities of these justifications are deferred to Appendix~\ref{app:merging}, where we also prove the following result of interest.

\begin{proposition}\label{prop:merging}
	Suppose that $v$ and $v'$ are two mergeable vertices in $\cT$. If $\cT^*$ is the trellis obtained from merging $v$ and $v'$, then $\cC(\cT^*)=\cC(\cT)$.
\end{proposition}

\begin{example}Consider again the trellis $\cT$ with $\cC(\cT)=\sS_3$.
	After merging, we obtain the trellis on the left, which we call $\cT_3$.
	\vspace{1mm}
	
	\begin{center}
		\begin{tikzpicture}[x=2cm,y=2cm]
			\tikzstyle{state}=[rectangle,fill=white,draw,line width=0.8mm]
			\tikzstyle{label}=[fill=white, inner sep=1pt]
			\node[state] at (0,0) (0) {$\varnothing$};
			\node[state] at (1,1) (1) {$1$};
			\node[state] at (1,0) (2) {$2$};
			\node[state] at (1,-1) (3) {$3$};
			\node[state] at (2,1) (12) {$12$};
			\node[state] at (2,0) (13) {$13$};
			\node[state] at (2,-1) (23) {$23$};
			\node[state] at (3,0) (123) {$123$};
			
			\path[-] (0) edge  node[label] {1} (1);
			\path[-] (0) edge  node[label] {2} (2);
			\path[-] (0) edge  node[label] {3} (3);
			
			\path[-] (1) edge  node[label] {2} (12);
			\path[-] (1) edge  node[pos=0.3, label] {3} (13);
			\path[-] (2) edge  node[pos=0.3, label] {1} (12);
			\path[-] (2) edge  node[pos=0.3, label] {3} (23);
			\path[-] (3) edge  node[pos=0.3, label] {1} (13);
			\path[-] (3) edge  node[label] {2} (23);
			
			\path[-] (12) edge  node[label] {3} (123);
			\path[-] (13) edge  node[label] {2} (123);
			\path[-] (23) edge  node[label] {1} (123);
		\end{tikzpicture}
		\hspace{20mm} 
		\begin{tikzpicture}[x=2cm,y=2cm]
			\tikzstyle{state}=[rectangle,fill=white,draw,line width=0.8mm]
			\tikzstyle{label}=[fill=white, inner sep=1pt]
			\node[state] at (0,0) (0) {$\varnothing$};
			\node[state] at (1,1) (1) {$1$};
			\node[state] at (1,0) (2) {$2$};
			\node[state] at (1,-1) (3) {$3$};
			\node[state] at (2,1) (12) {$12$};
			\node[state] at (2,0) (13) {$13$};
			\node[state] at (2,-1) (23) {$23$};
			\node[state] at (3,0) (123) {$123$};
			
			\path[-] (0) edge  node[label] {$a_{11}$} (1);
			\path[-] (0) edge  node[label] {$a_{21}$} (2);
			\path[-] (0) edge  node[label] {$a_{31}$} (3);
			
			\path[-] (1) edge  node[label] {$a_{22}$} (12);
			\path[-] (1) edge  node[pos=0.3, label] {$a_{32}$} (13);
			\path[-] (2) edge  node[pos=0.3, label] {$a_{12}$} (12);
			\path[-] (2) edge  node[pos=0.3, label] {$a_{32}$} (23);
			\path[-] (3) edge  node[pos=0.3, label] {$a_{12}$} (13);
			\path[-] (3) edge  node[label] {$a_{22}$} (23);
			
			\path[-] (12) edge  node[label] {$a_{33}$} (123);
			\path[-] (13) edge  node[label] {$a_{23}$} (123);
			\path[-] (23) edge  node[label] {$a_{13}$} (123);
		\end{tikzpicture}
	\end{center}
	\vspace{3mm}
	
	On the right is $\cT_3(\vA)$ where we relabel the edges in $\cT_3$ using the entries of $\vA$. This example can be generalized and henceforth, this trellis is referred to as the canonical trellis representation for $\vA$.
\end{example}

\begin{definition}[Canonical Trellis]\label{def:canonical}
	Fix $n$. Then the {\em canonical permutation trellis $\cT_n$} is defined as follows.
	\begin{itemize}
		\item (Vertices) For $0\le j\le n$, define $V_j$ to be the set of all $j$-subsets of $[n]$. Hence, $|V_j|=\binom{n}{j}$. So, $V=\bigcup_{j=0}^n V_j$ is the power set of $[n]$ and $|V|=2^n$.\\[-8mm]
		\item (Edges) For $j\in [n]$, we consider a pair $(u,v)\in V_{j-1}\times V_j$. Recall that $u$ and $v$ are $(j-1)$- and $j$-subsets, respectively. We have that $(u,v)$ is an edge if and only if $|v\setminus u|=1$.\\[-8mm]
		\item (Edge Labels) For an edge $(u,v)$, we have that $v\in V_j$ for some $j\in [n]$. Also, since $v\setminus u$~is~a~singleton, we set $i\in[n]$ such that $v\setminus u=\{i\}$. Then $L(u,v)=i$.
	\end{itemize}
	The {\em canonical trellis for $\vA$} is given by $\cT_n(\vA)$.
\end{definition}

It turns out that no two vertices in $\cT_n$ are mergeable and, in
fact, $\cT_n$ is the minimal trellis. The proof is a~straightforward
application of trellis theory and is therefore deferred to
Appendix\,\ref{app:minimal}.

\begin{proposition}\label{prop:minimal}
	$\cT_n$ is the minimal trellis for the set of all permutations $\sS_n$.
\end{proposition}

Next, we determine the number of operations incurred when we use $\cT_n(\vA)$ to compute $\per(\vA)$.

\begin{theorem}\label{thm:canonical}
	Let $\cT_n(\vA)=(V,E,L)$. Then
	$|V| = 2^n$ and $|E| = n2^{n-1}$.
	Therefore, $\per(\vA)$ can be computed using $n2^{n-1}-n$ multiplications and $(n-2)2^{n-1}+1$ additions with $\binom{n}{\floor{n/2}}$ space.
\end{theorem}

\begin{proof}
	The complexity measures follow directly from \eqref{eq:trellis-multiplies}, \eqref{eq:trellis-additions}, and 
	\eqref{eq:trellis-space}. Hence, it suffices to derive the graph-theoretic properties. The size of $V$ and the quantity $\max \{|V_j|: j\in [n]\}$ follow directly from the definition.
	For the number of edges, observe that the outdegree of a vertex in $V_j$ is $n-j$ for $0\le j\le n-1$.
	Therefore, we have that $|E|=\sum_{j=0}^{n-1}(n-j)\binom{n}{j}=\sum_{j=1}^{n}j\binom{n}{j}=n2^{n-1}$.
\end{proof}

Therefore, the number of arithmetic operations required by the permanent computation on $\cT_n(\vA)$ is of the same order as that required to evaluate
the Ryser formula \eq{Ryser} or the Glynn formula \eq{Glynn}.

\begin{remark} 
	In a study of codes with local permutation constraints, Sayir and
	Sarwar~\cite{Sayir.2015} proposed the use of $\cT_n$ to compute 
	the permanent of a matrix. Our work herein is independent from \cite{Sayir.2015}.
	In this work, we provide a detailed analysis
	of the number of arithmetic operations and also show that $\cT_n$ is the
	minimal trellis. Crucially, in the later sections, we use the merging
	procedure and other trellis manipulation techniques to dramatically
	reduce the number of arithmetic operations for certain structured
	matrices.
\end{remark}

\begin{remark} 
	The definition of mergeability and the merging procedure described in
	this section are slightly different from the ones given
	in \cite{Kschischang.1996, VardyK.1996}. This is because in the latter
	work, the authors are interested in preserving the {\em set} of paths
	without accounting for the multiplicities. In contrast, we are
	required to preserve the multiplicity for each path and hence, we
	provide a slightly different definition. As mentioned earlier, the
	correctness of the procedure is proved in
	Appendix~\ref{app:merging}. We also remark that the merging procedure
	mimics the construction of ordered binary decision diagrams (BDDs) for
	Boolean functions (see \cite{Bryant.1992, Bryant.1995} for a survey).
\end{remark}

\vspace{0.90ex}
\subsection{Reducing complexity via trellis normalization}

We propose a simple normalization technique that further reduces
the number of multiplications.
Let us fix $t\in [n]$ and normalize the $t$-th column of $\vA$. 
Specifically, we consider the matrix ${\vA|_t}$ such that
its $(i,j)$-th entry is $a_{ij}/a_{it}$.
Therefore, the $t$-th column of ${\vA|_t}$ consists of all ones%
\footnote{Here, we assume that $a_{it}\ne 0$ for all $i\in [n]$. In Remark~\ref{rem:normalization}, we describe how to define the normalized matrix when some entries in the $t$-th column is zero.}. 
Crucially, when we run~the~Viterbi algorithm on the corresponding
trellis $\cT_n\left({\vA|_t}\right)$, we need not perform any
multiplications to evaluate $\mu(v)$ for all $v\in V_t$. This is so
because for all $v\in V_t$, 
we have $\mu(v) = \sum_{(u,v)\in E}\mu(u)$.  Furthermore,
we recover the permanent of the original matrix $\vA$ by using the
fact that 
$$
\per(\vA) \,=\,
\left(\prod_{i=1}^n a_{it}\right)\per\left({\vA|_t}\right)
$$

Let us analyse the number of multiplications.
First, to normalize the matrix and obtain ${\vA|_t}$, we need $n(n-1)$ 
multiplications (recall that the $t$-th column is all ones by construction).
Next, we look at the number of multiplications in the trellis-based 
computation. 
The number of edges from $V_{t-1}$ to $V_t$ is $(n-t+1)\binom{n}{t-1}$,
and thus we save this quantity of multiplications when using the
normalized matrix ${\vA|_t}$. In other words, the number of
multiplications needed to compute the flow on $\cT_n\left({\vA|_t}\right)$
is $n2^{n-1}-n-(n-t+1)\binom{n}{t-1}$. 
Finally, we multiply $\per\left({\vA|_t}\right)$ by $a_{it}$ for all
$i\in [n]$, and this involves another $n$ multiplications.  Therefore,
in total, the number of multiplications is
$n2^{n-1}-(n-t+1)\binom{n}{t-1}+n^2-n$.

If we choose $t=\floor{n/2}+1$, we obtain the following theorem. 

\begin{theorem}\label{thm:normalization}
	Let $t=\floor{n/2}+1$. 
	Then computing $\per(\vA)$ on the trellis $\cT_n\left({\vA|_t}\right)$ invokes 
	\[
	n2^{n-1}-\ceil{n/2}\binom{n}{\floor{n/2}}+n^2-n \text{ multiplications and }
	(n-2)2^{n-1}+1  \text{ additions}.
	\]
\end{theorem}

We compare our trellis based approach with the state-of-the-art
methods of computing the permanent.  In \Tbref{table:general}, we
provide the exact number of operations required for Ryser's formula
and its variants. The careful derivation of the number of additions
and multiplications is provided in Appendix~\ref{app:ryserlike}.  From
the the table, we see that the best known prior work uses exactly
\[
M(n) \triangleq (n-1)2^{n-1}\text{ multiplications and } 
A(n) \triangleq (n+1)2^{n-1}+n^2-2n-1 \text{ additions.}
\] 
In contrast, the trellis based method uses $(n-2)2^{n-1}+1$ additions 
which is strictly less than $A(n)$~for~all~$n$.
Combining the trellis based method with normalization techniques
of this subsection, the number of multiplications is
$n2^{n-1}-\ceil{\frac n2}\binom{n}{\floor{n/2}}+n^2-n$. This
quantity is strictly less than $M(n)$ for $n\ge 7$.

\begin{remark}\label{rem:normalization}
	When $a_{it}=0$ for some values of $i\in [n]$, our trellis normalization techniques remain applicable. Specifically, we consider the matrix ${\vA|_t}$ whose $(i,j)$-th entry is $a_{ij}/a_{it}$ for all $i$ with $a_{it}\ne 0$. Then we have that the $(i,t)$-th entry of ${\vA|_t}$ is zero if $a_{it}=0$ and is one if $a_{it}\ne 0$.
	Then we proceed as before to compute $\per\left(\vA|_t\right)$ and we multiply the resulting flow by $\prod_{i\in [n], a_{it}\ne 0} a_{it}$ to recover $\per(\vA)$. It is straightforward to see that the number of multiplications and additions are bounded above by the values given in Theorem~\ref{thm:normalization}.
\end{remark}

\vspace{8mm}
\section{Matrices with repeated rows}
\label{sec:repeated}
\vspace{-0.50ex}

\looseness=-1
In this section, we compute the permanent for matrices with repeated
rows: we assume that $\vA$ has $t<n$ distinct rows and
these rows appear with multiplicities $m_1,m_2,\ldots, m_t$.
Specifically, we say that 
$\vA$ is a {\em repeated-row matrix} of type $\vm=m_1m_2\ldots m_t$ with rows
$\va_1,\va_2,\ldots \va_t$ if the row vector $\va_\ell=(a_{\ell
j})_{j\in [n]}$ appears exactly $m_\ell$ times for each $\ell\in
[t]$. Without loss of generality, we assume that $\vA$ is of the
following form:

\begin{center}
	\hspace*{7mm}
\begin{tabular}{b{5.40mm}cb{3mm}c}
	$\vA \mbox{$~=~$} $\newline \vspace*{15mm} &
\begin{tikzpicture}[mymatrixenv]
	\matrix[mymatrix] (m)  {
		\text{------} & \va_1 & \text{------}  \\[-2.70mm]
		\vdots & \vdots & \vdots \\
		\text{------} & \va_1 & \text{------}  \\
		\text{------} & \va_2 & \text{------}  \\[-2.70mm]
		\vdots & \vdots & \vdots \\
		\text{------} & \va_2 & \text{------}  \\
		\vdots & \vdots & \vdots \\
		\text{------} & \va_t & \text{------}  \\[-2.70mm]
		\vdots & \vdots & \vdots \\
		\text{------} & \va_t & \text{------}  \\
	};
	\mymatrixbraceleft{1}{3}{$m_1$}
	\mymatrixbraceleft{4}{6}{$m_2$}
	\mymatrixbraceleft{8}{10}{$m_t$}
\end{tikzpicture}
& \hspace*{-1.26ex}$=$\newline \vspace*{15mm} &
\begin{tikzpicture}[mymatrixenv]
	\matrix[mymatrix] (m)  {
		a_{11} & a_{12} & ~~~\cdots~~~ & a_{1n} \\[-2.70mm]
		\vdots & \vdots & \vdots & \vdots \\
		a_{11} & a_{12} & \cdots & a_{1n} \\
		a_{21} & a_{22} & \cdots & a_{2n} \\[-2.70mm]
		\vdots & \vdots & \vdots & \vdots \\
		a_{21} & a_{22} & \cdots & a_{2n} \\
		\vdots & \vdots & \vdots & \vdots \\
		a_{t1} & a_{t2} & \cdots & a_{tn} \\[-2.70mm]
		\vdots & \vdots & \vdots & \vdots \\
		a_{t1} & a_{t2} & \cdots & a_{tn} \\
	};
	\mymatrixbraceleft{1}{3}{$m_1$}
	\mymatrixbraceleft{4}{6}{$m_2$}
	\mymatrixbraceleft{8}{10}{$m_t$}
\end{tikzpicture}
\end{tabular}
\end{center}

Applying the merging procedure to $\cT_n(\vA)$ in the preceding section, we can reduce the number of vertices and edges to quantities polynomial in $n$ (when $t$ is constant). Indeed, for the part $V_1$, the vertices in $V_{11}=\{1,2,\ldots, m_1\}$ can be merged into one vertex as the past $\cP(v)=a_{11}$ for all $v$ in $V_{11}$.
Similarly, the vertices in $V_{21}=\{m_1+1,m_1+2,\ldots, m_1+m_2\}$  can be merged into a single vertex. 
Hence, for the vertices in $V_1$, we can merge these $n$ vertices into $t$ new vertices. 
Repeating this process for $V_2,V_3,\ldots, V_{n-1}$, we can reduce the number of vertices from $2^n$ to a quantity less than $n^t$, 
while the number of edges can be reduced from $n2^{n-1}$ to less than $tn^t$. 
Specifically, we obtain the following trellis (up to certain scaling).

\begin{definition}[Trellis for Repeated-Row Matrices]
	\label{def:repeated}
	Fix $n$ and let $\vA$ be a repeated-row matrix of type $\vm=m_1m_2\ldots m_t$ with rows $\va_1,\va_2,\ldots \va_t$. The trellis $\cT(\vA,\vm)$ is defined as follows:
	\begin{itemize}
		\item (Vertices)
		Define $V\triangleq \{\vlambda = (\lambda_\ell)_{\ell\in [t]} : 0\le \lambda_\ell \le m_\ell \text{ for }\ell\in [t]\}$. Hence, $|V|=\prod_{\ell=1}^t (m_\ell+1)$.
		For $0\le j\le n$, define $V_j=\{\vlambda\in V: \lambda_1+\lambda_2+\cdots +\lambda_t=j\}$. In other words, the vertices in $V_j$ consists of all integer-valued $t$-tuples whose entries sum to $j$.
		\item (Edges) For $j\in [n]$, we consider a pair $(\vmu,\vlambda)\in V_{j-1}\times V_j$. We place an edge $(\vmu,\vlambda)$ in $E$ if and only if there exists a unique $\ell^*$ such that $\lambda_{\ell^*}=\mu_{\ell^*}+1$ and $\lambda_{\ell}=\mu_{\ell}$ whenever $\ell\ne \ell^*$.
		\item (Edge Labels) For an edge $(\vmu,\vlambda)$, we have a unique $\ell^*$ such that the above condition hold. We then set the edge label $L(u,v)$ to be $a_{\ell^* j}$.
	\end{itemize}
\end{definition}

\begin{example}
	Let $n=6$ and $\vm=(1,2,3)$. 
	Suppose that $\vA$ is a repeated-row matrix of type $\vm$ with rows $\va_1,\va_2,\ldots \va_t$. 
	Then its corresponding trellis $\cT(\vA,\vm)$ is as follows. Here, we use colors to denote the labels. 
	For an edge from $V_{j-1}$ to $V_j$, the label of the edge is $a_{1j}$ if it is {\color{red}red}, $a_{2j}$ if it is {\color{green!50!black}green}, and $a_{3j}$ if it is {\color{blue}blue}.
	\vspace{2mm}

\begin{center}
	\small
	\begin{tikzpicture}[x=2cm,y=2cm]
	\tikzstyle{state}=[rectangle,fill=white,draw,line width=0.8mm]
	\tikzstyle{label}=[fill=white, inner sep=0pt]
	\node[state] at (0.00,0.00) (000) {$000$};
	\node[state] at (1.00,1.00) (100) {$100$};
	\node[state] at (1.00,0.00) (010) {$010$};
	\node[state] at (1.00,-0.60) (001) {$001$};
	\node[state] at (2.00,1.00) (110) {$110$};
	\node[state] at (2.00,0.40) (101) {$101$};
	\node[state] at (2.00,0.00) (020) {$020$};
	\node[state] at (2.00,-0.60) (011) {$011$};
	\node[state] at (2.00,-1.20) (002) {$002$};
	\node[state] at (3.00,1.00) (120) {$120$};
	\node[state] at (3.00,0.40) (111) {$111$};
	\node[state] at (3.00,-0.20) (102) {$102$};
	\node[state] at (3.00,-0.60) (021) {$021$};
	\node[state] at (3.00,-1.20) (012) {$012$};
	\node[state] at (3.00,-1.80) (003) {$003$};
	\node[state] at (4.00,0.40) (121) {$121$};
	\node[state] at (4.00,-0.20) (112) {$112$};
	\node[state] at (4.00,-0.80) (103) {$103$};
	\node[state] at (4.00,-1.20) (022) {$022$};
	\node[state] at (4.00,-1.80) (013) {$013$};
	\node[state] at (5.00,-0.20) (122) {$122$};
	\node[state] at (5.00,-0.80) (113) {$113$};
	\node[state] at (5.00,-1.80) (023) {$023$};
	\node[state] at (6.00,-0.80) (123) {$123$};
	
	\draw[red] (000)--(100);
	\draw[green] (000)--(010);
	\draw[blue] (000)--(001);
	\draw[green] (100)--(110);
	\draw[blue] (100)--(101);
	\draw[red] (010)--(110);
	\draw[green] (010)--(020);
	\draw[blue] (010)--(011);
	\draw[red] (001)--(101);
	\draw[green] (001)--(011);
	\draw[blue] (001)--(002);
	\draw[green] (110)--(120);
	\draw[blue] (110)--(111);
	\draw[green] (101)--(111);
	\draw[blue] (101)--(102);
	\draw[red] (020)--(120);
	\draw[blue] (020)--(021);
	\draw[red] (011)--(111);
	\draw[green] (011)--(021);
	\draw[blue] (011)--(012);
	\draw[red] (002)--(102);
	\draw[green] (002)--(012);
	\draw[blue] (002)--(003);
	\draw[blue] (120)--(121);
	\draw[green] (111)--(121);
	\draw[blue] (111)--(112);
	\draw[green] (102)--(112);
	\draw[blue] (102)--(103);
	\draw[red] (021)--(121);
	\draw[blue] (021)--(022);
	\draw[red] (012)--(112);
	\draw[green] (012)--(022);
	\draw[blue] (012)--(013);
	\draw[red] (003)--(103);
	\draw[green] (003)--(013);
	\draw[blue] (121)--(122);
	\draw[green] (112)--(122);
	\draw[blue] (112)--(113);
	\draw[green] (103)--(113);
	\draw[red] (022)--(122);
	\draw[blue] (022)--(023);
	\draw[red] (013)--(113);
	\draw[green] (013)--(023);
	\draw[blue] (122)--(123);
	\draw[green] (113)--(123);
	\draw[red] (023)--(123);
	\end{tikzpicture}
\end{center}
\vspace{2mm}

	If we perform the Viterbi algorithm on this trellis, 
	we have that $\mu({\rm toor})=\mu(123)$ to be the sum of 60 monomials of the form $\prod_{j\in [6]} a_{\ell_j j}$ with $\ell_j\in [3]$. Even though $\mu({\rm toor})$ does not correspond to $\per(\vA)$ (which is the summand of 720 monomials), we can recover the permanent by multiplying $\mu({\rm toor})$ with the scalar $m_1!m_2!m_3!=1!2!3!=12$. We also note that the trellis $\cT(\vA,\vm)$ has 24 vertices, while the canonical trellis $\cT_4$ has $2^6=64$ vertices. 
\end{example}

More generally, we have the following theorem which states that the flow $\mu$ at any vertex gives the permanent of some submatrix up to a certain scalar. 

\begin{theorem}\label{thm:repeated}
	Let $\vA$ be a repeated-row matrix of type $\vm=m_1m_2\ldots m_t$ with rows $\va_1,\va_2,\ldots \va_t$. Suppose that the Viterbi algorithm on $\cT(\vA,\vm)$ yields the flow $\mu(\vlambda)$ for each $\vlambda\in V\setminus\{\vzero\}$.
	If we set $j=\sum_{\ell=1}^t\lambda_\ell$ and $\vc_\ell=(a_{\ell 1}, a_{\ell 2}, \ldots, a_{\ell j})$ for $\ell\in [t]$,
	then $\per(\vC(\vlambda))=\lambda_1!\lambda_2!\cdots \lambda_t!\mu(\vlambda)$, where $\vC(\vlambda)$ is a repeated-row matrix of type $\vlambda$ with rows $\vc_1,\vc_2,\ldots, \vc_t$. Therefore, $\per(\vA)=\per(\vC(\vm))=m_1!m_2!\cdots m_t!\mu({\rm toor})$.
\end{theorem}

\begin{proof}
	We prove using induction on $j$. 
	When $j=1$ and $\vlambda$ has one on its $\ell$-th entry ($\ell\in [t]$), we have that $\vC(\vlambda)$ is the matrix $(a_{\ell 1})$ and we can easily verify that $\per(\vC(\vlambda))=a_{\ell 1}=1! \mu(\vlambda)$.
	
	Next, we assume that the hypothesis is true for some $j$ with $1\le j\le n$ and we prove the hypothesis for $j+1$. Consider $\vlambda$ with $\sum_{\ell=1}^t\lambda_\ell= j+1$. 
	For convenience, we show that $\per(\vC(\vlambda))=\lambda_1!\lambda_2!\cdots \lambda_t! \mu(\vlambda)$ for the case where all entries $\vlambda$ are strictly positive. The proof can extend easily to the case where some entry (or entries) is zero.
	
	For $\ell\in [t]$, let $\vmu_\ell\triangleq (\lambda_1,\ldots, \lambda_{\ell-1}, \lambda_{\ell}-1,\lambda_{\ell+1},\lambda_{t})$. Then $\vC(\vmu_\ell)$ is a repeated-row $j\times j$ matrix of type $\vmu_\ell$ and can be obtained from $\vC(\vlambda)$ by removing the row $\vb_\ell$ and the $(j+1)$-th column. 
	Using the Laplace expansion formula for permanents and the induction hypothesis, we have that 
	\begin{align*}
		\per(\vC(\vlambda)) 
		& = \sum_{\ell=1}^t \lambda_\ell a_{\ell,j+1} \per (\vC(\vmu_\ell))\\
		& = \sum_{\ell=1}^t \lambda_\ell (\lambda_1!\cdots\lambda_{\ell-1}!(\lambda_{\ell}-1)!\lambda_{\ell+1}!\cdots \lambda_t!)a_{\ell,j+1} \mu(\vmu_\ell)\\
		& = \lambda_1!\lambda_2!\cdots \lambda_t!\sum_{\ell=1}^t a_{\ell,j+1} \mu(\vmu_\ell)
	\end{align*} 
	It follows from the Viterbi algorithm that $\mu(\vlambda)=\sum_{\ell=1}^t a_{\ell,j+1} \mu(\vmu_\ell)$, completing the induction proof.
\end{proof}

Now, when $\vA$ is appropriately defined, it turns out that the scaled permanent value at each vertex corresponds to a certain probability event studied in order statistics. We describe this formally in the next section where we combine many of such trellises into one trellis with roughly the same number of vertices.
To end this section, we state explicitly the complexity measures of the trellis for repeated-row matrices.

\begin{theorem}
Let $\vA$ be a repeated-row matrix of type $\vm=m_1m_2\ldots m_t$ with
rows $\va_1,\va_2,\ldots \va_t$. Further let $\cT(\vA,\vm)=(V,E,L)$ be the
trellis constructed in Definition~\ref{def:repeated}. Then
$$
|V| \,=\, \prod_{\ell=1}^t(m_\ell+1)
\hspace{5.4ex}\text{and}\hspace{5.4ex}
|E| \,\le\, t\prod_{\ell=1}^t(m_\ell+1)
$$
Therefore $\per(\vA)$ can be computed using at most
$t(m_1+1)(m_2+1)\cdots(m_t+1)$ multiplications 
and at most $(t-1)(m_1+1)(m_2+1)\cdots(m_t+1)$ additions.
\end{theorem}

\begin{proof}
	The size of $V$ follows directly from Definition~\ref{def:repeated}. For the number of edges, since each vertex in $V$ has degree at most $t$, we have that $|E|\le t|V|$. The complexity measures then follow from \eqref{eq:trellis-multiplies} and \eqref{eq:trellis-additions}.
\end{proof}

As before, we compare the trellis-based approach with the best known
exact method of computing~permanents for repeated-row matrices. This
method is due to Clifford-Clifford \cite{Clifford.2020} and it is
based on~the~following inclusion-exclusion formula:
\begin{equation}\label{eq:clifford2}
	\per(\vA) = 
	(-1)^n\sum_{r_1=0}^{m_1}\cdots\sum_{r_t=0}^{m_t}
	(-1)^{r_1+\cdots+r_t}
	\left(\prod_{\ell=1}^{t}\binom{m_\ell}{r_\ell}\right)
	\prod_{j=1}^{n}\sum_{\ell=1}^{m}r_\ell a_{\ell j}
\end{equation}
We have that \eqref{eq:clifford2} invokes $(n-1)\left[\prod_{\ell=1}^t(m_\ell+1)-1\right]$ 
multiplications and $(n+1)\left[\prod_{\ell=1}^t(m_\ell+1)-2\right]$ additions. 
We defer the detailed derivation of this to Appendix~\ref{app:ryserlike}.
Observe that the number of arithmetic operations is reduced
by a factor of about $n/t$ when we use the trellis $\cT(\vA,\vm)$ to compute
the permanent.

\section{Order statistics}
\vspace{-0.50ex}

In this section, we adapt the trellis defined in Definition~\ref{def:repeated} to efficiently compute a certain joint probability distribution in order statistics. 

Formally, suppose that we have $n$ independent real-valued random variables $X_1,X_2,\ldots, X_n$. We draw one sample from each population distribution and order them so that 
$X_{(1)}\le X_{(2)}\le \ldots\le X_{(n)}$.
Fix $t$ distinct integers with $1\le r_1<r_2<\cdots <r_t\le n$ and $t$ real values with $x_1\le x_2\le \cdots x_t$. We have the following formula \cite{Vaughan.1972,Bapat.1989}:
\begin{align} 
&\prob{\bigwedge_{\ell=1}^t X_{(r_\ell)}\le x_\ell} = 
\sum_{i_t=r_t}^{n}
\sum_{i_{t-1}=r_{t-1}}^{i_t}\cdots
\sum_{i_1=r_1}^{i_2} F(i_1,i_2,\ldots, i_t), \notag\\
&\text{\hspace{50mm}where }
F(i_1,i_2,\ldots, i_t) = \frac{\per(\vB(i_1,i_2,\ldots, i_t))}{i_1!(i_2-i_1)!\cdots (i_t-i_{t-1})! (n-i_t)!}.\label{eq:order}
\end{align} 
Here, $\vB(i_1,i_2,\ldots,i_t)$ is a matrix whose rows are obtained from one of the following $t+1$ possibilities: $\vb_1,\vb_2,\ldots, \vb_{t+1}$. For $\ell\in [t+1]$,  the row vector $\vb_\ell=(b_{\ell j})_{j\in [n]}$ is defined by the population distributions and the values $x_1, x_2,\ldots, x_t$. Specifically, 
\begin{equation*}
b_{\ell j} = \begin{cases}
\prob{X_j\le x_1} & \text{if } \ell=1,\\
\prob{x_{\ell-1}<X_j\le x_\ell} & \text{if } 2\le \ell\le t,\\
\prob{X_j> x_t} & \text{if } \ell=t+1.
\end{cases}
\end{equation*}
The multiplicities of each row or the type of $\vB(i_1,i_2,\ldots, i_t)$ is determined by the $i_\ell$'s. Specifically, set 
$m_1=i_1$, $m_{t+1}=n-i_t$ and $m_\ell=i_\ell-i_{\ell-1}$ for $2\le \ell\le t$.
Then $\vB(i_1,i_2,\ldots, i_t)$ is a repeated-row matrix of type $\vm=(m_\ell)_{\ell\in [t+1]}$ with $\vb_1,\vb_2,\ldots, \vb_{t+1}$.

Prior to this work, for fixed $t$, polynomial-time methods to compute \eqref{eq:order} were only known when the random variables $X_1,X_2,\ldots, X_n$ were drawn from at most two variables \cite{Glueck.2008}.
Now, since $\vB(i_1,i_2,\ldots, i_t)$ has at most $t+1$ distinct rows, we can apply either  Clifford-Clifford or the trellis-based method to compute each permanent in $O(n^{t+2})$ or $O(n^{t+1})$ time, respectively.
However, as there are $\Theta(n^t)$ permanents in the formula \eqref{eq:order}, this naive approach has running time $O(n^{2t+2})$ (Clifford-Clifford) or $O(n^{2t+1})$ (trellis-based).

Now, if we apply our merging technique to all the trellises constructed, it turns out that we are able to compute \eqref{eq:order} with {\em only one trellis} that has at most $n^{t+1}$ vertices! Specifically, the trellis is defined below.

\begin{definition}[Trellis for Order Statistics]
	Given $n$ population distributions and $x_1,x_2,\ldots, x_t$, we define the row vectors $\vb_1,\vb_2,\ldots, \vb_{t+1}$ as above. 
	We also have a $t$-tuple $\vr=(r_\ell)_{\ell\in [t]}$. The trellis $\cT^{\rm o}(\vB,\vr)$, is defined as follows:
	\vspace{-2mm}
	
	\begin{itemize}
		\item (Vertices)
		Define 
		$V\triangleq \{\vlambda = (\lambda_\ell)_{\ell\in [t+1]} : 0\le \lambda_\ell \le n-r_{\ell-1} \text{ for }\ell\in [t+1]\}$, where we set $r_0=0$.
		 Hence, $|V|=\prod_{\ell=1}^{t+1} (n-r_{\ell-1}+1)\le n^{t+1}$.
		As before, for $0\le j\le n$, define $V_j=\{\vlambda\in V: \lambda_1+\lambda_2+\cdots +\lambda_{t+1}=j\}$. 
		\item (Edges) For $j\in [n]$, we consider a pair $(\vmu,\vlambda)\in V_{j-1}\times V_j$. We place an edge $(\vmu,\vlambda)$ in $E$ if and only if there exists a unique $\ell^*$ such that $\lambda_{\ell^*}=\mu_{\ell^*}+1$ and $\lambda_{\ell}=\mu_{\ell}$ whenever $\ell\ne \ell^*$.
		\item (Edge Labels) For an edge $(\vmu,\vlambda)$, we have a unique $\ell^*$ such that the above condition hold. We then set the edge label $L(u,v)$ to be $a_{\ell^* j}$.
	\end{itemize}
\end{definition}

If we run the Viterbi algorithm on the trellis for ordered statistics $\cT^{\rm o}(\vB,\vr)$, it follows from Theorem~\ref{thm:repeated} that the flow $\mu(\vlambda)$ of the vertex $\vlambda$ in $V_n$ is $F(i_1,i_2,\ldots, i_t)$ where $i_\ell=\sum_{s=1}^{\ell} \lambda_s$ for $\ell\in [t] $. To compute $\prob{\bigwedge_{\ell=1}^t X_{(\ell)}\le x_\ell}$, we consider the set of vertices $V(\vr)=\{\vlambda \in V_n : \sum_{s=1}^{\ell} \lambda_s\ge r_\ell \text{ for all } \ell\in [t]\}$ and 
add the flows of these vertices.
That is, we have that
$\prob{\bigwedge_{\ell=1}^t X_{(\ell)}\le x_\ell}
= \sum_{\vlambda\in V(\vr)} \mu(\vlambda)$.
In this final step, we need at most $|V|$ additions, and
we summarize our discussion with the following theorem.

\begin{theorem}
	Given $n$ population distributions and $x_1,x_2,\ldots, x_t$, we define the row vectors $\vb_1,\vb_2,\ldots, \vb_{t+1}$ as above. 
	We also have a $t$-tuple $\vr=(r_\ell)_{\ell\in [t]}$. 
	Then the joint probability in \eqref{eq:order} can be computed with at most
	$(t+1)\prod_{\ell=1}^{t+1} (n-r_{\ell-1}+1)\le (t+1)n^{t+1}$ multiplications and at most $(t+1)\prod_{\ell=1}^{t+1} (n-r_{\ell-1}+1)\le (t+1)n^{t+1}$ additions.
\end{theorem}

\section{Sparse matrices}
\label{sec:sparse}
\vspace{-0.50ex}

In this section, we consider sparse matrices, or, matrices with few nonzero entries.
Specifically, we fix an integer $d<n$ and set $p=d/n$ and $q=1-p$.
We consider a random $n\times n$-matrix $\vA$ where each entry is nonzero with probability $p$ and zero with probability $q$.
In other words, $\vA$ has on average $d$ nonzero entries in each row and column.
As most entries in $\vA$ are zero, we observe that most vertices in the canonical trellis ${\cal T}_n$ are {\em non-essential}, meaning that they do not lie on any path from the root to the toor.
Such vertices (and all edges incident to them) can be pruned away without
affecting the flow from the root to the toor. In the following, we formally describe this pruning procedure.


\begin{definition}[Sparse Trellis]\label{def:sparse}
	Let $\vA$ be an $n\times n$ matrix. Then the {\em sparse trellis $\cT^{\rm s}_n(\vA)=(V,E,L)$} is defined to be the trellis resulting from the following construction.\\[1mm]
	\hspace*{20mm}Set $V_0=\{\varnothing\}$\\[1pt]
	\hspace*{20mm}for $j\in [n]$ \\[1pt]
	\hspace*{25mm}for $u\in V_{j-1}$ \\[1pt]
	\hspace*{30mm}for $i \in\{\ell\in [n]: a_{\ell j}\ne 0, \ell\notin u\}$\\[1pt]
	\hspace*{35mm}Set $v\gets u\cup\{i\}$\\[1pt]
	\hspace*{35mm}Add $v$ to $V_j$\\[1pt]
	\hspace*{35mm}Add the edge $(u,v)$ to $E$ with label $L(u,v)\gets a_{ij}$
\end{definition}

As before, to evaluate the trellis complexity, we estimate the expected number of vertices and edges in $\cT^{\rm s}_n(\vA)$.
First, we observe that the degree of each vertex in $V_{j-1}$ is at most the number of nonzero entries in column $j$ of $\vA$ for $j\in [n]$.
Since the expected number of nonzero entries in column is $d$, we have that expected number of edges is at most $d$ times the expected number of vertices.
Therefore, it remains to provide an upper bound on the expected number of vertices.

\begin{lemma}\label{lem:sparse}
	Let $d\le n$ and set $q=1-d/n$.
	Define $U(n)$ as in \eqref{Un}.
	Then the expected number of vertices in $\cT^{\rm s}_n(\vA)$ is at most 
	$U(n)\le\phi_T^n$, where $\phi_T = 2-e^{-d}$.
\end{lemma}

Before we provide the proof of Lemma~\ref{lem:sparse}, we use it with  \eqref{eq:trellis-multiplies} and \eqref{eq:trellis-additions} to obtain upper bounds on the number of multiplications and additions.
Observe that the estimate in the following theorem demonstrates that the trellis-based approach provides an exponential speedup of Ryser's formula when the matrix is sparse.

\newpage

\begin{theorem}\label{thm:sparse}
	Fix $d<n$ and set $p=d/n$ and $q=1-p$.
	Let $\vA$ be a random $n\times n$ matrix $\vA$ where each entry is nonzero with probability $p$ and zero with probability $q$.
	Let $U(n)$ be as defined in \eqref{Un}. 
	Then computing $\per(\vA)$ on $\cT^{\rm s}_n(\vA)$, on average, invokes at most $d U(n)\le d\phi_T^n$ multiplications and at most $(d-1)U(n)\le (d-1)\phi_T^n$ additions.
\end{theorem}

For the rest of this section, we prove Lemma~\ref{lem:sparse}.
To this end, we have the following characterization of when a vertex appears in the trellis $\cT^{\rm s}_n(\vA)$.

\begin{proposition}
	Let $\vA$ be an $n\times n$ matrix.
	Suppose that $v$ be a nonempty $j$-subset of $[n]$.
	We consider the $j\times j$ submatrix $\vA(v)$ whose columns are those indexed by $[j]$ and rows are those indexed by $v$.
	Then $v$ is a vertex in $\cT^{\rm s}_n(\vA)$ if and only if $\per(\vA(v))$ is nonzero.
\end{proposition}

Hence, we proceed to estimate the probability of when a random $j\times j$ matrix has a nonzero permanent.
Specifically, let $\vA$ be a random $j\times j$ matrix and
we consider the following random events.
\vspace{-2mm}
\begin{align*}
\mathbb{A}_j & = \text{event where } \per(\vA) \text{ is non-zero},\\
\mathbb{B}_j & = \text{event where all rows and columns in } \vA \text{ are non-zero},\\
\mathbb{C}_j & = \text{event where all columns in } \vA \text{ are non-zero}.
\end{align*}
Here, a row (or a column) is nonzero if it contains some nonzero entry.
Now, the event $\mathbb{A}_j$ implies the event~$\mathbb{B}_j$, which in turn implies the event $\mathbb{C}_j$. Hence, $\prob{\mathbb{A}_j}\le \prob{\mathbb{B}_j}\le \prob{\mathbb{C}_j}$ and our task is to determine the probabilities of the latter two events.

Now, for the event $\mathbb{B}_j$, we observe that for any $k$-subset $I$ of the rows, 
the probability that the rows in $I$ are nonzero is $q^k\sum_{\ell=0}^{j-k}\binom{j-k}{\ell}p^\ell q^{j-k-\ell}=(q^k-q^j)^j$. 
Then using the principle of inclusion-exclusion, we have that $\prob{\mathbb{B}_j}=\sum_{k=0}^j(-1)^k\binom{j}{k}(q^k-q^j)^j$. 
On the other hand, for the event $\mathbb{C}_j$, we simply have that $\prob{\mathbb{C}_j}=(1-q^j)^j$. 

Finally, we proceed to complete the proof of Lemma~\ref{lem:sparse}.
So, for each $j$-subset $v$, the probability that $v$ is a vertex in the trellis is $\prob{\mathbb{A}_j}$. Hence, the expected number of vertices in $V_j$ is $\binom{n}{j}\prob{\mathbb{A}_j}$ and
by linearity of expectation, the expected number of vertices in the entire trellis is $1+\sum_{j=1}^n\binom{n}{j}\prob{\mathbb{A}_j}$.
Using the event $\mathbb{B}_j$, the expected number of vertices is at most $1+\sum_{j=1}^n\binom{n}{j}\sum_{k=0}^j(-1)^k\binom{n}{j}\binom{j}{k}(q^j-q^k)^j$, which is $U(n)$.
Using the event $\mathbb{C}_j$, we have that $U(n)$ is at most 
\[\sum_{j=0}^n\binom{n}{j}(1-q^j)^j
\le \sum_{j=0}^n\binom{n}{j}(1-q^n)^j
= (1-q^n)^n\le (2-e^{-d})^n.
\]

This completes the proof of Lemma~\ref{lem:sparse}.

\begin{remark}
Our analysis follows that in Erd\"os and Renyi's seminal paper \cite{Erdos.1964}.
In the paper, Erd\"os and Renyi provided the conditions for a random matrix to have a nonzero permanent with high probability. 
In \cite{Erdos.1964}, the inclusion-exclusion formula for event $\mathbb{B}_j$ was determined and used to estimate event $\mathbb{A}_j$ (see also Stanley~\cite{Stanley.2011}). However, as we were unable to obtain a closed formula for $U(n)$, we turn to event $\mathbb{C}_j$ to obtain the expression $\phi_T$.
\end{remark}

As mentioned in Section~\ref{sec:intro-sparse}, similar exponential speedup of the Ryser's formula was achieved by a few authors \cite{Servedio.2005,Bjorklund.2012,Lundow.2020}.
In the same section, we also discussed the sparsity assumptions of the various works. In Table~\ref{table:sparse}, we compare the number of multiplications and observe that the value of $\phi_T$ is smaller than $\phi_1$ and $\phi_3$ but larger than $\phi_2$ for most values of $d$. 
Nevertheless, if we numerically compute the value $\phi_U = \lim_{n\to\infty} U(n)^{1/n}$, we see that the value of $\phi_U$ is significantly less than $\phi_1,\phi_2,\phi_3$. 
Indeed, for the case $d=3$ and for matrix dimensions up to 50, we plot the number of multiplications for the various methods in Figure~\ref{fig-sparse} and we observe that the expected number of multiplications for the trellis-based method is exponentially smaller than the state-of-the-art methods.

\section{Traveling salesperson problem}
\label{sec:tsp}
\vspace{-0.50ex}

In this section, we extend the Viterbi algorithm to compute functions that resembles the permanent. 
Specifically, we study the traveling salesperson problem (TSP). 
Applying standard trellis manipulation techniques to the canonical permutation trellis, we obtain a trellis that represents the collection of TSP tours. 
Interestingly, using the trellis to solve the TSP instance recovers the Held-Karp algorithm \cite{HeldKarp.1962} -- the best known exact method for solving TSP.

Formally, we consider $n$ cities represented by $[n]$ and
let $\vD=(d_{ij})_{1\le i,j\le n}$ be a distance matrix%
\footnote{Here, we set $d_{ii}=0$ for $i\in [n]$.} 
with $d_{ij}$ being the distance from City $i$ to City $j$.
We define a {\em travelling salesperson (TSP) tour} to be a string $\vx=x_1x_2x_3\cdots x_n x_{n+1}$ of length $n+1$ such that $x_1=x_{n+1}=1$ and $x_2x_3\cdots x_n$ is a permutation over $\{2,3,\ldots, n\}$.
The length of a TSP tour $\vx$ is given by the sum $\ell(\vx)=\sum_{i\in [n]}d_{x_i x_{i+1}}$ and 
the TSP problem is to determine $\min \{\ell(\vx):\, \vx \text{ is a TSP tour}\}$.

Let ${\cal TSP}(n)$ denote the set of all TSP tours for $n$ cities. To simplify our exposition, we consider the case $n=4$. Modifying the canonical trellis defined in Definition~\ref{def:canonical} for the alphabet $\{2,3,4\}$, we obtain the following trellis $\cT$ that represents the code ${\cal TSP}(n)$. Here, we use colors to represent the labels. Black, {\color{red}red}, {\color{green!50!black}green}, and  {\color{blue}blue} edges are labeled 1, 2, 3, and 4, respectively.

\begin{center}
	\begin{tikzpicture}[x=2cm,y=2cm]
		\tikzstyle{state}=[rectangle,fill=white,draw,line width=0.8mm]
		\tikzstyle{label}=[fill=white, inner sep=0pt]
		\node[state] at (-1,0) (0) {root};
		\node[state] at (0,0) (01) {$\varnothing$};
		\node[state] at (1,1) (2) {$\{2\}$};
		\node[state] at (1,0) (3) {$\{3\}$};
		\node[state] at (1,-1)(4) {$\{4\}$};
		\node[state] at (2,1) (23) {$\{2,3\}$};
		\node[state] at (2,0) (24) {$\{2,4\}$};
		\node[state] at (2,-1)(34) {$\{3,4\}$};
		\node[state] at (3,0) (234) {$\{2,3,4\}$};
		\node[state] at (4,0) (41) {toor};
		
		\draw (0)--(01);
		\draw[red] (01)--(2);
		\draw[green] (01)--(3);
		\draw[blue] (01)--(4);
		\draw[green] (2)--(23);
		\draw[blue]  (2)--(24);
		\draw[red]  (3)--(23);
		\draw[blue] (3)--(34);
		\draw[red] (4)--(24);
		\draw[green] (4)--(34);
		\draw[blue]   (23)--(234);
		\draw[green] (24)--(234);
		\draw[red]  (34)--(234);
		
		\draw(234)--(41);			
	\end{tikzpicture}
\end{center}

Hence, given a distance matrix $\vD$, it remains to relabel the paths from ${\rm root}$ to ${\rm toor}$ so that the sum of edge distances corresponds to the length of the TSP tour.
Unfortunately, this turns out to be not possible.
For example, if we consider the TSP tour 12341, its fourth edge traverses from $\{2,3\}$ to $\{2,3,4\}$ in the trellis and its distance should correspond to $d_{34}$. 
However, if we consider the TSP tour 13241, its fourth edge {\em also} traverses from $\{2,3\}$ to $\{2,3,4\}$ in the trellis and now, we require this distance to be $d_{24}$. 
Hence, it is not possible to relabel the edges in the trellis $\cT$ with distances so that the lengths of all TSP tours are consistent.

To resolve this problem, 
we consider another collection of words ${\cal W}=\{\vx\in [4]^5: x_1=x_5=1,$ 
$x_2,x_3,x_4\in \{2,3,4\}, x_2\ne x_3, x_3\ne x_4\}$. 
In other words, if we consider a complete graph on $[5]$, then ${\cal W}$ is the set of all walks starting and ending at $1$ and whose intermediate vertices belong to $\{2,3,4\}$.
Then the trellis $\cT'$ below that represents the set of these walks $\cal W$. Moreover, the trellis $\cT'$ admits a relabeling such that its path lengths corresponds to the distance matrix. 
Here, the label colors are as before.

\begin{center}
	\begin{tikzpicture}[x=2cm,y=2cm]
		\tikzstyle{state}=[rectangle,fill=white,draw,line width=0.8mm]
		\tikzstyle{label}=[fill=white, inner sep=0pt]
		\node[state] at (-1,0) (0) {root};
		\node[state] at (0,0) (01) {$1$};
		\node[state] at (1,1) (12) {$2$};
		\node[state] at (1,0) (13) {$3$};
		\node[state] at (1,-1)(14) {$4$};
		\node[state] at (2,1) (22) {$2$};
		\node[state] at (2,0) (23) {$3$};
		\node[state] at (2,-1)(24) {$4$};
		\node[state] at (3,1) (32) {$2$};
		\node[state] at (3,0) (33) {$3$};
		\node[state] at (3,-1) (34) {$4$};
		\node[state] at (4,0) (41) {$1$};
		
		\draw (0)--(01);
		\draw[red] (01)--(12);
		\draw[green] (01)--(13);
		\draw[blue] (01)--(14);
		\draw[green] (12)--(23);
		\draw[blue]  (12)--(24);
		\draw[red]   (13)--(22);
		\draw[blue]  (13)--(24);
		\draw[red]   (14)--(22);
		\draw[green] (14)--(23);
		
		\draw[green] (22)--(33);
		\draw[blue]  (22)--(34);
		\draw[red]   (23)--(32);
		\draw[blue]  (23)--(34);
		\draw[red]   (24)--(32);
		\draw[green] (24)--(33);
		
		\draw(32)--(41);\draw(33)--(41);\draw(34)--(41);			
	\end{tikzpicture}
\end{center}
Then using $\vD$, we simply relabel the edge $(i,j)$ with $i,j\in [4]$ with the distance $d_{ij}$. We observe that the length of any path in the relabeled $\cT'$ can be obtained from the sum of all the edge labels. 
The question is then: how do we ``combine'' the trellises $\cT$ and $\cT'$? 

To do so, we turn to trellis theory and look at the {\em intersection} of the two trellises. The intersection of two trellises is first described in \cite{Kschischang.1996} and later formally introduced and studied in \cite{KoetterV.2003}.

\begin{definition}\label{def:tsp}
	Let $\cT=(V,E,L)$ and $\cT'=(V',E',L')$ be two trellises of the same length $n$ and the same label alphabet $\Sigma$. The {\em intersection trellis $\cT\cap \cT'=(V^*,E^*,L^*)$} of $\cT$ and $\cT'$ is defined as follows.
	\begin{itemize}
	\item (Vertices) For $0\le j\le n$, we set $V^*_j= V_j\times V'_j$.
	\item (Edges) For $j\in [n]$, we consider a pair $((u,u'), (v,v'))\in V_{j-1}^*\times V_{j}^*$. We have that $((u,u'), (v,v'))$ is an edge if and only if $(u,v)\in E$, $(u',v')\in E'$ and $L(u,v)=L'(u',v')$.
	\item (Edge Labels) For any edge $((u,u'), (v,v'))$ in $E^*$, we have that $(u,v)$ and $(u',v')$ are both edges in $\cT$ and $\cT'$, respectively, and $L(u,v)=L(u',v')$. Then we set $L^*((u,u'), (v,v'))=L(u,v)$.
	\end{itemize}
\end{definition}

\begin{theorem}
	Let $\cT=(V,E,L)$ and $\cT'=(V',E',L')$ be two trellises of the same length $n$ and the same label alphabet $\Sigma$. 
	Then $\cC(\cT)\cap \cC(\cT')=\cC(\cT\cap\cT')$.
\end{theorem}

\begin{example}Consider the trellises $\cT$ and $\cT'$ as before. Then $\cT\cap \cT'$ is the following trellis which represents ${\cal TSP}(4)$.
	Here, we omitted the vertex $({\rm root},{\rm root})$.

\begin{center}
	\begin{tikzpicture}[x=2.5cm,y=2cm]
		\tikzstyle{state}=[rectangle,fill=white,draw,line width=0.5mm]
		\tikzstyle{label}=[fill=white, inner sep=0pt]
		\node[state] at (0,0) (01) {$\varnothing$,1};
		\node[state] at (1,1) (2) {$\{2\},2$};
		\node[state] at (1,0) (3) {$\{3\},3$};
		\node[state] at (1,-1)(4) {$\{4\},4$};
		\node[state] at (2,1.2) (232) {$\{2,3\},2$};
		\node[state] at (2,0.8) (233) {$\{2,3\},3$};
		\node[state] at (2,0.2) (242) {$\{2,4\},2$};
		\node[state] at (2,-0.2) (244) {$\{2,4\},4$};
		\node[state] at (2,-0.8) (343) {$\{3,4\},3$};
		\node[state] at (2,-1.2) (344) {$\{3,4\},4$};
		\node[state] at (3,-1)(2342) {$\{2,3,4\},2$};
		\node[state] at (3,0)(2343) {$\{2,3,4\},3$};
		\node[state] at (3,1) (2344) {$\{2,3,4\},4$};
		\node[state] at (4,0) (41) {toor,$1$};
		
		\draw[red] (01)--(2);
		\draw[green] (01)--(3);
		\draw[blue] (01)--(4);
		\draw[green] (2.east)--(233.west);
		\draw[blue]  (2.east)--(244.west);
		\draw[red]   (3.east)--(232.west);
		\draw[blue]  (3.east)--(344.west);
		\draw[red]   (4.east)--(242.west);
		\draw[green]  (4.east)--(343.west);
		\draw[blue]  (232.east)--(2344.west);
		\draw[blue]  (233.east)--(2344.west);
		\draw[green]  (242.east)--(2343.west);
		\draw[green]  (244.east)--(2343.west);		
		\draw[red]  (343.east)--(2342.west);
		\draw[red]  (344.east)--(2342.west);
		\draw(2344.east)--(41);\draw(2343.east)--(41);\draw(2342.east)--(41);	
	\end{tikzpicture}
\end{center}

Given a distance matrix $\vD$, we relabel the edges in $\cT\cap \cT'$ in the following manner: for the edge 
$((u,i),(v,j))$ where $u,v\subseteq \{2,3,4\}$ and $i,j\in [4]$, we label this edge with the distance $d_{ij}$. Then with this relabeling, we can verify that the TSP tours 12341 and 13241 has lengths $d_{12}+d_{23}+d_{34}+d_{41}$ and $d_{13}+d_{32}+d_{24}+d_{41}$, respectively.
This is as desired.
\end{example}

Proceeding for general $n$, we obtain the following trellis that solves a TSP instance exactly.

\begin{definition}[Trellis for Traveling Salesperson Problem]
	Consider $n$ cities represented by $[n]$.
	Let $\vD=(d_{ij})_{1\le i,j\le n}$ be a distance matrix with $d_{ij}$ being the distance from City $i$ to City $j$.
	The trellis $\cT^{\rm TSP}(\vD)$, is defined as follows:
	\begin{itemize}
		\item (Vertices)
		For $2\le j\le n$, let 
		$V_j \triangleq \{(S,v): S\subseteq \{2,3,\ldots, n\},\, |S|=j+1,\, v\in S\}$.
		Define $V_1=\{(\varnothing,1)\}$ and $V_{n+1}=\{({\rm toor}, 1)\}$.
		\item (Edges) 
		For $j\in [n]$, we consider a pair $((S_1,u),(S_2,v))\in V_{j-1}\times V_j$ with $|S_1|=j$ and $|S_2|=j+1$. 
		We place an edge $((S_1,u),(S_2,v))$ in $E$ if and only if $S_2\setminus S_1=\{v\}$. 
		We also include $((\{2,3\ldots, n\},i),({\rm toor}, 1))$ in $E$ for all $i\in \{2,3,\ldots, n\}$.
		\item (Edge Labels) For an edge $((S_1,i),(S_2,j))$, we label it with $d_{ij}$.
	\end{itemize}
\end{definition}

To compute the length of the shortest TSP tour, we modify the Viterbi algorithm described in Section~\ref{sec:canonical} and apply it on the trellis $\cT^{\rm TSP}(\vD)$.

For each vertex $v\in V$, we assign a flow $\mu(v)$ that is computed in the following recursive manner.\\[2mm]
	\hspace*{20mm}Set $\mu((\varnothing,1))=0$.\\[1pt]
	\hspace*{20mm}for $j\in \{2,3,\ldots,n+1\}$ \\[1pt]
	\hspace*{30mm}for $v\in V_j$ \\[1pt]
	\hspace*{40mm}Set $\mu(v) = \min_{(u,v)\in E}L(u,v)+\mu(u)$
\vspace{2mm}
	
Then the length of a shortest TSP tour is given by $\mu(({\rm toor},1))$. 
As in the previous sections, we can use the sizes of $V$ and $E$ to provide explicit numbers of comparisons and additions in this computation.

\begin{proposition}\label{prop:tsp}
	Let $\vD$ be an $n\times n$ distance matrix that defines a TSP instance.
	Let $\cT_{\rm TSP}(\vD)=(V,E,L)$ be the trellis constructed in Definition~\ref{def:tsp}.
	Then $|V|=(n-1)2^{n-2}+2$ and $|E|=(n-1)(n-2)2^{n-3}+2(n-1)$. Therefore, the TSP instance can be solved exactly using $|E|-(n-1)=(n-1)(n-2)2^{n-3}+(n-1)$ additions and $|E|-|V|+1=(n-1)(n-4)2^{n-3}+(2n-3)$ comparisons.
\end{proposition}

Running the Viterbi algorithm on the trellis $\cT_{\rm TSP}(\vD)$ recovers the  well-known Held-Karp dynamic program for TSPs \cite{HeldKarp.1962}. 
In their seminal paper, Held and Karp also counted the number of additions and comparisons. 
While the number of additions corresponds to our derivation in Proposition~\ref{prop:tsp}, the number of comparisons given in \cite{HeldKarp.1962} is incorrect and instead corresponds to the number given in Proposition~\ref{prop:tsp}.

As discussed in Section~\ref{sec:into-tsp}, this section illustrates that trellises can be used to compute certain ``permanent-like'' functions (cf. \eqref{per*-def}).
Specifically, given a distance matrix $\vD$, we construct the trellis $\cT_{\rm TSP}(\vD)$ and then the Viterbi algorithm yields the flow value: 
\begin{equation*}
	{\textstyle\min_{\,\vx\in{\cal TSP}(n)}} \sum_{i=1}^n d_{x_ix_{i+1}}~=~
	{\textstyle\min_{\,\vsigma\in\raisebox{-0.36ex}{\LARGE$\circ$}_n}} \sum_{i=1}^n d_{i\sigma_i},
	\vspace{-.90ex}
\end{equation*}
which solves the TSP problem. 
Notably, we achieve this by intersecting the canonical permutation
trellis $\cT_n$ with another natural trellis.
Even though we did not improve on the state-of-the-art, we expect that
similar tools can be applied to other problems in order to obtain
trellises for other permanent-like functions.

\section*{Acknowledgements}

The authors would like to thank Fedor Petrov and the other
contributors at {\tt mathoverflow.net}~for~suggesting the proof of the
inequality given in \Lref{lem:sparse}.

\appendix

\section{Merging is correct}
\label{app:merging}

In this appendix, we demonstrate that the merging procedure described in Section~\ref{sec:canonical} is correct. 
Formally, suppose that two vertices $v,v'\in \cT$ are mergeable according to \eqref{eq:merge} and that $\cT^*$ is the trellis obtained by merging $v$ and $v'$. Then in Proposition~\ref{prop:merging}, we are required to show that $\cC(\cT)=\cC(\cT^*)$.

To this end, for the label alphabet $\Sigma$, we consider the collection of all nonempty strings of finite length $\Sigma^+$. Then under the usual binary operation of string concatenation, the set $\Sigma^+$ forms a {\em semigroup}. 
We next consider the field of rational numbers $\mathbb{Q}$ and 
the {\em semigroup algebra} $\mathbb{Q}[\Sigma^+]$. That is, $\mathbb{Q}[\Sigma^+]$ is the set of all formal expressions
\begin{equation*}
	\sum_{\vsigma\in \Sigma^+} a_\vsigma \vsigma, \text{ where }a_\vsigma\in \mathbb{Q}.
\end{equation*}
Here, $a_\vsigma$ is assume to be zero for all but a finite set of $\vsigma$'s. In other words, the sum is defined only for a finite subset of $\Sigma^+$.

Then it is a standard exercise to show that the following algebraic properties hold. Let $a\in \mathbb{Q}$ and $\vrho, \vsigma,\vtau\in \mathbb{Q}[\Sigma^+]$. 
First, the associative law applies: $\vrho(\vsigma\vtau)=(\vrho\vsigma)\vtau$ and $a(\vsigma\vtau)=\vsigma(a\vtau)$.
Next, the distributive law applies: $\vrho(\vsigma+\vtau)=\vrho\vsigma+\vrho\vtau$ and $a(\vsigma+\vtau)=a\vsigma+a\vtau$.

Now, for any trellis $\cT$, since $\cC(\cT)$ is a multiset of strings, we can regard it as an element in $\mathbb{Q}[\Sigma^+]$ and we can perform the algebraic operations according to the above laws. We are now ready to prove Proposition~\ref{prop:merging}.

\begin{proof}[Proof of Proposition~\ref{prop:merging}]
	First, for convenience, we extend the domain of $L$ from $E$ to the set of  $V\times V$. 
	Specifically, we set $L(u,u')=0$ if $(u,u')$ is not an edge.
	Then any path that passes through the zero label is assigned to the zero element in $\mathbb{Q}[\Sigma^+]$ and we see that $\cC(\cT)$ is preserved.
	Also, the merging rule can be simplified as such:
	\begin{itemize}
		\item $L^*(w,w')=L(w,w')$ if $w\ne v$, $w'\ne v$, $w\ne v'$, and $w'\ne v'$,
		\item $L^*(w,v^*)=L(w,v)+L(w,v')$, and 
		\item $L^*(v^*,w')=(L(v,w')+L(v',w'))/2$.
	\end{itemize}

	Next, we observe that a path $\vp$ in $\cT$ belongs to $\cT^*$ if and only if both $v$ and $v'$ are not on the path $\vp$. 
	Hence, to show that $\cC(\cT)=\cC(\cT^*)$, it suffices to show that the paths through the merged node $v^*$ in $\cT^*$ is the sum of the paths through $v$ and $v'$ in $\cT$. In other words,  $\cP(v^*)\cF(v^*)=\cP(v)\cF(v)+\cP(v')\cF(v')$.
	Indeed, 
	\begin{align*}
		2\cP(v^*)\cF(v^*) 
		& = 2\left(\sum_{u}\cP(u)L(u,v^*)\right)\left(\sum_{w}L(v^*,w)\cF(w)\right)\\
		& = 2\left(\sum_{u}\cP(u)(L(u,v)+L(u,v'))\right)\left(\sum_{w}\frac12(L(v,w)+L(v',w))\cF(w)\right)\\
		& = \left(\sum_{u}\cP(u)L(u,v)\right)\left(\sum_{w}L(v,w)\cF(w)\right) + \left(\sum_{u}\cP(u)L(u,v')\right)\left(\sum_{w}(L(v,w)\cF(w)\right) \\
		& ~~+\left(\sum_{u}\cP(u)L(u,v)\right)\left(\sum_{w}L(v',w)\cF(w)\right) + \left(\sum_{u}\cP(u)L(u,v')\right)\left(\sum_{w}L(v',w)\cF(w)\right)\\
		&= \cP(v)\cF(v)+\cP(v')\cF(v)+\cP(v)\cF(v')+ \cP(v')\cF(v') \\
		&= 2\cP(v)\cF(v)+2\cP(v')\cF(v').
	\end{align*}
	Here, the penultimate equality follows from the mergeability condition~\eqref{eq:merge}. Dividing both sides by two, we obtained the desired equation.
\end{proof}

\section{Minimal trellises}
\label{app:minimal}

In this section, we prove the minimality of the trellis $\cT_n$ (Proposition~\ref{prop:minimal}).
To this end, we require the notion of a {\em biproper trellis} and a {\em rectangular code}.

\begin{definition}\label{def:biproper}
	A trellis is {\em proper} if the edges leaving any node are labeled distinctly, while a trellis is {\em co-proper} if the edges entering any node are labeled distinctly. 
	A trellis is {\em biproper} if it is both proper and co-proper.
\end{definition}

\begin{definition}\label{def:rectangular}
	Let $\cC$ be a collection of words over $\Sigma$ of the same length.
	Then $\cC$ is a {\em rectangular code} if it has the following property:
	\begin{equation*}
		\va\vc, \va\vd, \vb\vc \in \cC \text{ implies that } \vb\vd \in \cC.
	\end{equation*}
\end{definition}

Under certain mild conditions, Vardy and Kschischang demonstrated the following key result that establishes the equivalence of biproper and minimal trellises \cite{VardyK.1996}.

\begin{theorem}[Vardy, Kschischang\cite{VardyK.1996}]\label{thm:rectangular}
	Let $\cT$ be a trellis such that $\cC(\cT)$ is a rectangular code.
	Then $\cT$ is minimal if and only if $\cT$ is biproper.
\end{theorem}

Recall that the canonical permutation trellis $\cT_n$ is a trellis representation for the set of permutations $\cS_n$. Following Theorem~\ref{thm:rectangular}, to establish the minimality of $\cT_n$, it suffices to show that $\cS_n$ is rectangular and that $\cT_n$ is biproper.
Before we proceed with the proof, we note that Kschischang has proved that $\cS_n$ is rectangular \cite{Kschischang.1996}. 
Specifically, he showed that $\cS_n$ belongs to a subclass of rectangular codes, known as {\em maximal fixed-cost codes}.
To keep our exposition self-contained, we directly prove that $\cS_n$ is rectangular without defining the class of maximal fixed-cost codes.

\begin{proof}[Proof of Proposition~\ref{prop:minimal}]
	We first demonstrate that $\cS_n$ is rectangular. Suppose that $\va\vc$, $\va\vd$, and $\vb\vc$ be words belonging to $\cS_n$. Then set $X$ to be the set of symbols in $D$, that is, $X\triangleq\{\sigma: \sigma\in\vc\}$.
	Since $\va\vc$ and $\vb\vc$ are both permutations, the set of symbols in $\va$ is equal to the set of symbols in $\vb$ and corresponds to $[n]\setminus X$. Now, since $\vb\vc$ is a permutation, we have that the set of symbols in $\vc$ is $X$. Therefore, we conclude that $\vb\vc$ is a permutation, establishing that $\cS_n$ is rectangular.
	
	Next, we show that $\cT_n$ is biproper. 
	Let $v$ be a node in $\cT_n$ with $|v|=j$. 
	Then $v$ has $n-j$ outgoing edges that are distinctly labeled by the $n-j$ symbols in $[n]\setminus v$. Therefore, $\cT_n$ is proper.
	On the other hand, there are $j$ edges entering $v$ and they are distinctly labeled by the $j$ symbols in $v$. So, $\cT_n$ is co-proper, as desired.
\end{proof}

\section{Arithmetic complexity of Ryser formula and its variants}
\label{app:ryserlike}

In this appendix, we provide a detailed  derivation of the exact number of additions and multiplies for the various permanent formulae given in this paper (specifically, Table~\ref{table:general} and Section~\ref{sec:repeated}).
Similar analysis that estimates these numbers has appeared in earlier works \cite{NW.1978, Glynn.2010}. Here, we present a careful derivation for the {\em exact} number of additions and multiplications. For convenience, we replicate the Ryser's formula here.

\begin{equation}
	\per(\vA) = (-1)^n\sum_{S\subseteq [n],\, S\ne\varnothing} (-1)^{|S|} \prod_{j=1}^n \sum_{i\in S}a_{ij}.
\end{equation}

In this formula \eqref{Ryser}, we have $2^n-1$ summands and each summand involves a product of $n$ terms. Hence, the total number of multiplications for Ryser's formula is simply $(n-1)(2^n-1)$.

Next, we derive the number of additions using a naive approach. 
For any nonempty subset $S$ of size $k$, we need $(k-1)$ additions to compute the term $\sum_{i\in S}a_{ij}$ for each $j\in [n]$. 
Therefore, the total number of such additions is 
$\sum_{k=1}^n n(k-1)\binom{n}{k}=n(n-2)2^{n-1}+n$.
Finally, we need to add these $2^n-1$ terms together and thus, the total number of additions is $(n^2-2n+2)2^{n-1}+n-2$\,.

Later, Nijenhuis and Wilf \cite{NW.1978} used Gray codes to reduce the number of additions by a factor of $n$. Specifically, the nonempty subsets of $[n]$ can be arranged in the order $S_1,S_2,\ldots, S_{2^n-1}$ such that $|S_1|=1$ and $|S_j\setminus S_{j-1}|=1$ or $|S_{j}-1\setminus S_{j}|=1$ for $j\ge 2$. Then we use $n$ ``sums'' to store the value $\sum_{i\in S}a_{ij}$ for each $j$ and update these sums according to the order $S_1,S_2,\ldots, S_{2^n-1}$. Since there are only $n$ additions in each step, the total number of additions is $(n+1)(2^n-2)$ (this includes the final addition of the $2^n-1$ summands).

We summarize this discussion in Table~\ref{table:general}. In the table, we also include the exact number of arithmetic operations for Nijenhuis-Wilf and Glynn's formulae, whose derivations appear in the following subsections.

\subsection{Nijenhuis-Wilf formula}

Nijenhuis and Wilf \cite[Chapter 23]{NW.1978} proposed the following formula to compute permanents. For $j\in [n]$, let $C_j\triangleq -\frac12 \sum_{i=2}^n a_{ij}$. In other words, $C_j$ is the ``negative half'' of the sum of entries in column $j$.

\begin{equation}\label{eq:nw}
	\per(\vA) = (-1)^{n} 2 \sum_{S\subseteq [n-1]} (-1)^{|S|} \prod_{j=1}^n (C_j +\sum_{i\in S}a_{ij}).
\end{equation}

Since the formula~\eqref{eq:nw} involves $2^{n-1}$ summands,  we require $2^{n-1}(n-1)$ multiplies to compute.

As for the number of additions, we used a Gray order and arrange the subsets $S_1,S_2,\ldots, S_{2^{n-1}}$ as before so that ``neighboring subsets" differ in only one element. 
For $S_1$, we choose it to be the empty subset and we need to compute $C_j$ for $j\in [n]$, invoking $n(n-1)$ additions.
Then, for each subsequent subset, we add $n$ numbers to current $n$ ``sum''s. Hence, the number of additions is $n$ for the subsequent subsets. 
Finally, we have another $2^{n-1}-1$ additions to add the products.
Therefore, the total number of additions is $(n+1)(2^{n-1}-1)+n(n-1)$.

\subsection{Glynn's formula}

Let $\Delta(n)$ denote the set of vectors over $\pm 1$ of length $n$ starting with $1$. Using a certain polarization identity, Glynn proposed the following formula to compute permanents \cite{Glynn.2010}.

\begin{equation}\label{eq:glynn}
	\per(\vA) = (1/2^{n-1}) \sum_{\vdelta\subseteq \Delta(n)} \left(\prod_{k=1}^n\delta_k\right) \prod_{j=1}^n \sum_{i=1}^n \delta_i a_{ij}.
\end{equation}

As with the Nijenhuis-Wilf formula, Glynn's formula~\eqref{eq:glynn} involves $2^{n-1}$ summands and hence, requires $2^{n-1}(n-1)$ multiplies to compute. 

As for the number of additions, we order the sequences in $\Delta(n)$ using a Gray order: $\vdelta_1,\vdelta_2,\ldots,\vdelta_{2^{n-1}}$. That is, the Hamming distance between $\vdelta_i$ and $\vdelta_{i+1}$ is exactly one for $i\in [2^{n-1}-1]$. 
Then performing the same analysis as Nijenhuis-Wilf, we have that
the number of additions is $(n+1)(2^{n-1}-1)+n(n-1)$.

\subsection{Clifford-Clifford formula}

Here, we assume that the matrix has repeated rows. 
Specifically, we assume that $\vA$ is a repeated-row matrix of type $\vm$ with rows $\va_1,\va_2,\ldots, \va_t$ (see Section~\ref{sec:repeated}).

Methods to compute permanents for repeated-row matrices have been studied for the application of Boson Sampling. 
The following formula appeared in \cite[Appendix D]{Shchesnovich.2013} and
using generalized Gray codes, Clifford and Clifford (2020) provided a faster method to evaluate the formula. 

Let $R(\vm)=\{\vr\in\mathbb{Z}^t\setminus\{\vzero\}:\, 0\le r_\ell\le m_\ell,\, \ell\in [t]\}$.
Then we can rewrite \eqref{eq:clifford2} as follow.

\begin{equation*}
	\per(\vA) = 
	(-1)^n\sum_{\vr\in R(\vm)}
	(-1)^{r_1+\cdots+r_t}
	\left(\prod_{\ell=1}^{t}\binom{m_\ell}{r_\ell}\right)
	\prod_{j=1}^{n}\sum_{\ell=1}^{m}r_\ell a_{\ell j}
\end{equation*}

Since $|R_m|=[(m_1+1)(m_2+1)\cdots(m_t+1)-1]$, we have that the number of multiplies is $(n-1)[(m_1+1)(m_2+1)\cdots(m_t+1)-1]$.

As before,  we can arrange the vectors in  $R(\vm)$ in a generalized Gray code order: $\vr_1,\vr_2,\ldots$. That is, the norm between consecutive vectors is one, i.e. $|\vr_i-\vr_{i+1}|=1$ for all $i$. 
Then performing the same analysis as before, we have that
the total number of additions is $(n+1)(|R_m|-1)=(n+1)[(m_1+1)(m_2+1)\cdots(m_t+1)-2]$.



\end{document}